\documentclass{article}

\usepackage{microtype}
\usepackage{graphicx}
\usepackage{subfigure}
\usepackage{booktabs} % for professional tables

\usepackage{hyperref}

\usepackage[accepted]{icml2023}

\usepackage{amsmath}
\usepackage{amssymb}
\usepackage{mathtools}
\usepackage{amsthm}

\usepackage[capitalize,noabbrev]{cleveref}

\usepackage{bm}
\usepackage{bbm}
\newcommand{\indep}{\perp \!\!\! \perp}

\theoremstyle{plain}
\newtheorem{theorem}{Theorem}[section]

\theoremstyle{definition}
\newtheorem{definition}[theorem]{Definition}
\newtheorem{assumption}[theorem]{Assumption}
\theoremstyle{remark}

\usepackage[textsize=tiny]{todonotes}

\icmltitlerunning{Proper Scoring Rules for Survival Analysis}

\begin{document}

\twocolumn[
\icmltitle{Proper Scoring Rules for Survival Analysis}



\icmlsetsymbol{equal}{*}

\begin{icmlauthorlist}
\icmlauthor{Hiroki Yanagisawa}{ibm}
\end{icmlauthorlist}

\icmlaffiliation{ibm}{IBM Research - Tokyo, Tokyo, Japan}

\icmlcorrespondingauthor{Hiroki Yanagisawa}{yanagis@jp.ibm.com}

\icmlkeywords{Scoring rule, survival analysis, time-to-event analysis, censoring}

\vskip 0.3in
]



\printAffiliationsAndNotice{}  % leave blank if no need to mention equal contribution

\begin{abstract}
Survival analysis is the problem of estimating probability distributions for future event times, which can be seen as a problem in uncertainty quantification. Although there are fundamental theories on strictly proper scoring rules for uncertainty quantification, little is known about those for survival analysis. In this paper, we investigate extensions of four major strictly proper scoring rules for survival analysis and we prove that these extensions are proper under certain conditions, which arise from the discretization of the estimation of probability distributions. We also compare the estimation performances of these extended scoring rules by using real datasets, and the extensions of the logarithmic score and the Brier score performed the best.
\end{abstract}

\section{Introduction}

The theory of {\em scoring rules} is a fundamental theory in statistical analysis, and it is widely used in uncertainty quantification (see, e.g.,~\cite{mura2008,parmigiani2009,Ben10,STW15}).  Suppose that there is a random variable $Y$ whose cumulative distribution function (CDF) is $F_{Y}$.  Given an estimation $\hat{F}_{Y}$ of $F_{Y}$ and a single sample $y$ obtained from $Y$, a scoring rule $S(\hat{F}_{Y}, y)$ is a function that returns an evaluation score for $\hat{F}_{Y}$ based on $y$.  Since $\hat{F}_{Y}$ is a CDF and $y$ is a single sample of $Y$, it is not straightforward to choose an appropriate scoring rule $S(\hat{F}_{Y}, y)$.  The theory of scoring rules suggests how to choose an appropriate one.  This theory proves that a certain class of scoring rules satisfies this natural property: the average evaluation score $S(\hat{F}_{Y}, y)$ over $y \sim Y$ is minimized only by the true CDF $F_{Y}$.  A scoring rule that satisfies this property is called {\em strictly proper} in this theory.  Examples of strictly proper scoring rules include the pinball loss, the logarithmic score, the Brier score, and the ranked probability score (see, e.g.,~\cite{GR07} for the definitions of these scoring rules).  In uncertainty quantification, it is standard to use a strictly proper scoring rule both for a loss function to train machine learning models and for an evaluation metric to evaluate the models.  Note that, if we use a non-proper scoring rule $S(\hat{F}_{Y}, y)$ as a loss function, a prediction model (e.g., a neural network model) might find an estimation $\hat{F}_{Y}$ such that $S(\hat{F}_{Y}, y) < S(F_{Y}, y)$ holds for $y \sim Y$ on average and such $\hat{F}_{Y}$ could be very different from true $F_{Y}$.

{\em Survival analysis}, which is also known as {\em time-to-event analysis}, can be seen as a problem in uncertainty quantification.  Despite the long history of research from the seminal work~\cite{Cox72} on survival analysis (see, e.g.,~\cite{WLR19} for a comprehensive survey), little is known about the strictly proper scoring rules for survival analysis.  Therefore, we investigate extensions of these strictly proper scoring rules for survival analysis.

In survival analysis, the time between a well-defined starting point and the occurrence of an event is called the {\em survival time} or {\em event time}, and the goal of survival analysis is to estimate the probability distribution of event times.  In healthcare applications, an event usually corresponds to an undesirable event for a patient (e.g., a death or the onset of disease).  Survival analysis has important applications in many fields such as credit scoring~\cite{DCB17} and fraud detection~\cite{ZYW19} as well as healthcare.  

Datasets for survival analysis are {\em censored}, which means that events of interest might not be observed for a number of data points. This may be due to either a limited observation time window or missing traces caused by other irrelevant events.  In this paper, we consider only {\em right censored} data, which is a widely studied problem setting in survival analysis.  The exact event time of a right censored data point is unknown; we know only that the event had not happened up to a certain time for the data point.  The time between a well-defined starting point and the last observation time of a right censored data point is called the {\em censoring time}.

Many neural network models have been proposed for survival analysis (e.g.,~\cite{ADZJSN19,KW21,PJJZ22}).  A common problem with these models is that they define their own custom loss functions, and they use these loss functions without proving that they are strictly proper in terms of the theory of scoring rules.  Indeed, Rindt et al.~\yrcite{RHSS22} show that the loss functions used in~\cite{ADZJSN19,KW21} are not proper.  Moreover, survival models have been evaluated by custom evaluation metrics without proving that these metrics are proper in terms of the theory of scoring rules.  Popular metrics used for survival analysis include the integrated Brier score~\cite{GSSS99} and variants of C-index~\cite{ABB05,UCPDW11}.  However, all of them are not proper~\cite{BKG18,RHSS22}.  We also note that Sonabend et al.~\yrcite{SBV22} discuss the problems of using these variants of C-index as evaluation metrics in survival analysis.

The only exception to the above argument is~\cite{RHSS22}.  This paper shows a rigorous proof that an extension of the logarithmic score for survival analysis is strictly proper.  Note that this paper is not the first one that uses this extension of the logarithmic score (e.g.,~\cite{LZYS18,RQZYZQY19,THW21}).  However, it is usually  used in {\em part} of the loss functions of the proposed models, and these loss functions are used without the proof of properness.

\paragraph{Our contributions.}
We analyze survival analysis through the lens of the theory of scoring rules.  First, we prove that Portnoy's estimator~\cite{Por03}, which is an extension of the pinball loss for survival analysis, is proper under certain conditions.  This is the first proof for the properness for Portnoy's estimator.  In addition, we show such conditions are due to the discretization of the estimation of a probability distribution and we explain why such conditions are required to be proper scoring rules for survival analysis.  Second, we show that the proof of strict properness of the extension of the logarithmic score~\cite{RHSS22} is based on implicit assumptions by showing its alternative proof.  Third, we show two new proper scoring rules for survival analysis under certain conditions by extending the Brier score and the ranked probability score. These scoring rules are the first scoring rules with rigorous proofs of properness as extensions of the Brier score and the ranked probability score.  Finally, we compare these four extensions of the scoring rules by using real datasets, and the results show that the extensions of the logarithmic score and Brier score performed the best.

\section{Related Work}

Survival analysis has been traditionally studied under the {\em proportional hazard assumption}.  Its seminal work is the Cox model~\cite{Cox72}, and many other prediction models have been proposed under this strong assumption.  Since outputs of these models are scalar values called {\em hazard rates} and are not CDFs, we use different types of loss functions and evaluation metrics in traditional survival analysis.  One of the popular evaluation metrics is the concordance index (C-index)~\cite{HCPLR82}, which is a generalization of the Kendall rank correlation coefficient.  See, e.g.,~\cite{WLR19} for a comprehensive survey on survival analysis based on this assumption.  In this paper, we focus on survival analysis {\em without} this assumption.

We note that there are many loss functions used in survival models that can be seen as variants of known scoring rules.  
\begin{itemize}
\item {\bf Pinball loss.}  Portnoy's estimator~\cite{Por03}, which is an extension of the pinball loss, has been used in quantile regression-based survival analysis~\cite{Por03,NBP06,PJJZ22}.  It was unknown if this estimator is proper or not in terms of the theory of scoring rules, and we are the first to prove that this estimator is proper under a certain condition.
\item {\bf Brier score.}   The IPCW Brier score~\cite{GSSS99} and integrated Brier score~\cite{GSSS99} are widely used in survival analysis (e.g.,~\cite{KBS19,HHDG20,HGPWPR21,ZMW21}) as variants of the Brier score.  However, Rindt et al.~\yrcite{RHSS22} show that neither of them is proper in terms of the theory of scoring rules.
\item {\bf Ranked probability score.}  Variants of the ranked probability score have been proposed in \cite{ADZJSN19,KW21}, but Rindt et al.~\yrcite{RHSS22} show that they are not proper in terms of the theory of scoring rules.
\end{itemize}

\section{Preliminaries}

Given a feature vector $x \in X$, let $T_{x}$ and $C_{x}$ be random variables for the event time and censoring time of $x$, respectively.  Unless otherwise stated, we consider a fixed $x$, and we denote them by $T$ and $C$ by omitting the subscript $x$ from $T_{x}$ and $C_{x}$, respectively.

Let $t \sim T$ and $c \sim C$ be samples obtained from $T$ and $C$, respectively.  We assume that $t$ and $c$ are positive real values (i.e., $t \in \mathbb{R}^{+}$ and $c \in \mathbb{R}^{+}$).  In survival analysis, we can observe only the minimum $z = \min \{ t, c \}$, and we use $\delta = \mathbbm{1}(t \leq c)$ to indicate whether $z$ represents the true event time (i.e., $\delta=1$ means $z$ is uncensored and $z=t$) or $z$ represents the censoring time (i.e., $\delta=0$ means $z$ is censored and $z=c$).  In this paper, a pair of samples $(t,c)$ is often represented as a pair of values $(z, \delta)$ to emphasize that we can observe only one of $t$ and $c$.  We assume that we have prior knowledge that $z$ is at most $z_{\max}$ (i.e., $0 < z \leq z_{\max}$ holds for any $z$).  Let $F(t)$ be the CDF of $T$, which is defined as $F(t) = \mathrm{Pr}(T \leq t)$.   By the definition of $F(t)$, we have $F(0) = 0$, and we can represent the probability that the true event time is between $t_{1}$ and $t_{2}$ as $\mathrm{Pr}(t_{1} < T \leq t_{2}) = F(t_{2}) - F(t_{1})$.
  
Survival analysis is the problem of estimating the $\hat{F}(t)$ of the true CDF $F(t)$.  For simplicity, we assume that both $F(t)$ and $\hat{F}(t)$ are monotonically increasing continuous functions.  This means that $F(t_{1}) < F(t_{2})$ holds if and only if $0 \leq t_{1} < t_{2} < \infty$.  This assumption enables us to calculate $F(t)$ for any time $0 \leq t < \infty$ and to calculate $F^{-1}(\tau)$ for any quantile level $0 \leq \tau \leq 1$.  When we estimate $\hat{F}(t)$ by using a prediction model (e.g., a neural network), we usually discretize $p = \hat{F}(t)$ along the $p$-axis or the $t$-axis as shown in Fig.~\ref{fig:uncertainity_quantification}.  In quantile regression-based survival analysis, $p=\hat{F}(t)$ is discretized along the $p$-axis, $\hat{F}^{-1}(\tau_{i})$ is estimated for $0 = \tau_{0} < \tau_{1} < \cdots < \tau_{B-1} < \tau_{B} = 1$, and we assume that $\hat{F}^{-1}(\tau_{0}) = 0$ and $\hat{F}^{-1}(\tau_{B}) = z_{\max}$.   In distribution regression-based survival analysis, $p=\hat{F}(t)$ is discretized along the $t$-axis, $\hat{F}(\zeta_{i})$ is estimated for $0 = \zeta_{0} < \zeta_{1} < \cdots < \zeta_{B-1} < \zeta_{B} = z_{\max}$, and we assume that $\hat{F}(\zeta_{0})=0$ and $\hat{F}(\zeta_{B})=1$.

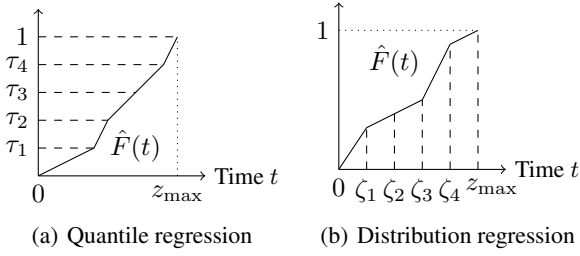
\begin{figure}[t]
  \subfigure[Quantile regression]{
\centering
\begin{tikzpicture}[scale=1.85]
\draw (0,0)--(0.4,0.2)--(0.5,0.4)--(0.7,0.6)--(0.9,0.8)--(1.0,1.0);
\draw[dashed] (0,0.2)--(0.4,0.2);
\draw[dashed] (0,0.4)--(0.5,0.4);
\draw[dashed] (0,0.6)--(0.7,0.6);
\draw[dashed] (0,0.8)--(0.9,0.8);
\draw[dashed] (0,1)--(1,1);
\node at (0,0) [below] {$0$}; 
\node at (0,0.2) [left] {$\tau_{1}$};
\node at (0,0.4) [left] {$\tau_{2}$};
\node at (0,0.6) [left] {$\tau_{3}$};
\node at (0,0.8) [left] {$\tau_{4}$};
\node at (0,1.0) [left] {$1$};
\draw [->] (0,0)--(0,1.2);
\draw [->] (0,0)--(1.2,0);
\node at (0.7,0.4) [below] {$\hat{F}(t)$};
\node at (1.2,0) [right] {\small Time $t$};
\draw[dotted] (1,0)--(1,1);
\node at (1,0) [below] {$z_{\max}$};
\end{tikzpicture}
  }
  \subfigure[Distribution regression]{
\centering
\begin{tikzpicture}[scale=1.85]
\draw (0,0)--(0.2,0.3)--(0.4,0.4)--(0.6,0.5)--(0.8,0.9)--(1.0,1.0);
\draw[dashed] (0.2,0)--(0.2,0.3);
\draw[dashed] (0.4,0)--(0.4,0.4);
\draw[dashed] (0.6,0)--(0.6,0.5);
\draw[dashed] (0.8,0)--(0.8,0.9);
\draw[dashed] (1,0)--(1,1);
\node at (0,0) [below] {$0$}; 
\node at (0.2,0) [below] {$\zeta_{1}$};
\node at (0.4,0) [below] {$\zeta_{2}$};
\node at (0.6,0) [below] {$\zeta_{3}$};
\node at (0.8,0) [below] {$\zeta_{4}$};
\node at (1.1,0) [below] {$z_{\max}$};
\draw [->] (0,0)--(0,1.2);
\draw [->] (0,0)--(1.2,0);
\node at (0,1.0) [left] {$1$};
\node at (0.4,0.6) [above] {$\hat{F}(t)$};
\node at (1.2,0) [right] {\small Time $t$};
\draw[dotted] (0,1)--(1,1);
\end{tikzpicture}
  }
\caption{Two types of discretization of probability distribution $\hat{F}(t)$ with $B=5$.}
\label{fig:uncertainity_quantification}
\end{figure}

Throughout this paper, we assume that the censoring time and the event time are independent of each other given a feature vector $x$.  This assumption is usually represented as follows.
\begin{assumption}\label{assumption:independence}
$T \indep C | X$.
\end{assumption}
This assumption is widely used in survival analysis (e.g., the classical Kaplan-Meier estimator~\cite{KM58} and the calibration metric D-calibration~\cite{HHDG20}).  We can find examples of the other stronger assumptions (e.g., unconditionally random right censoring) used in survival analysis in~\cite{Peng21}.

\section{Proper Scoring Rules for Survival Analysis}\label{sec:scoring_rule}

We briefly review the theory of scoring rules for uncertainty quantification.  Let $Y$ be a random variable, and let $F_{Y}(y)$ be its CDF, which is defined as $F_{Y}(y) = \mathrm{Pr}(Y \leq y)$. A {\em scoring rule} is a function $S(\hat{F}_{Y}, y)$ that returns a real value (i.e., an evaluation score) for inputs $\hat{F}_{Y}$ and $y$, where $\hat{F}_{Y}$ is an estimation of $F_{Y}$ and $y$ is a sample obtained from $Y$.  In this paper, we consider {\em negatively-oriented} scoring rules, which means that a smaller score is better.  We can interpret the scoring rule $S(\hat{F}_{Y}, y)$ as a penalty function for the misestimation of $\hat{F}_{Y}$ for a sample $y$.

The {\em proper} and {\em strictly proper} scoring rules are defined as follows.

\begin{definition}\label{def:original_proper}
A scoring rule $S(\hat{F}_{Y}, y)$ is {\em proper} if
\[ \mathop{\mathbb{E}}_{y \sim Y} [S(\hat{F}_{Y}, y)] \geq \mathop{\mathbb{E}}_{y \sim Y} [S(F_{Y}, y)]. \]
\end{definition}

\begin{definition}\label{def:original_strict_proper}
A scoring rule $S(\hat{F}_{Y}, y)$ is {\em strictly proper} if \[ \mathop{\mathbb{E}}_{y \sim Y} [S(\hat{F}_{Y}, y)] \geq \mathop{\mathbb{E}}_{y \sim Y} [S(F_{Y}, y)] \]
holds and the equality holds only when $\hat{F}_{Y} = F_{Y}$.
\end{definition}

These definitions are based on a natural property that any scoring rule should satisfy.  Definition~\ref{def:original_strict_proper} means that we can recover the true $F_{Y}$ by minimizing the average evaluation score $S(\hat{F}_{Y}, y)$ over $y \sim Y$ for a strictly proper scoring rule $S(\cdot,\cdot)$.

We extend these definitions of the proper and strictly proper scoring rules for survival analysis.   We define the {\em proper} and {\em strictly proper} scoring rules for survival analysis by changing the inputs of a scoring rule $S(\hat{F}, (z,\delta))$ from $F_{Y}$ and $y$ to $F$ and $(z,\delta)$.

\begin{definition}\label{def:survival_proper}
A scoring rule $S(\hat{F}, (z,\delta))$ is {\em proper} if
\[ \mathop{\mathbb{E}}_{(t,c) \sim (T,C)} [S(\hat{F}, (z,\delta))] \geq \mathop{\mathbb{E}}_{(t,c) \sim (T,C)} [S(F, (z,\delta))]. \]
\end{definition}

\begin{definition}\label{def:survival_strict_proper}
A scoring rule $S(\hat{F}, (z,\delta))$ is {\em strictly proper} if
\[ \mathop{\mathbb{E}}_{(t,c) \sim (T,C)} [S(\hat{F}, (z,\delta))] \geq \mathop{\mathbb{E}}_{(t,c) \sim (T,C)} [S(F, (z,\delta))] \]
holds and the equality holds only when $\hat{F} = F$.
\end{definition}

Following the standard approach of using a strictly proper scoring rule in uncertainty quantification~\cite{BHW22}, we explain how to use a scoring rule $S(\hat{F}, (z,\delta))$ as a loss function in survival analysis.  Given a training dataset $\{ (x^{(i)}, z^{(i)}, \delta^{(i)}) \}_{i=1}^{n}$, we formulate survival analysis as minimizing the empirical loss
\[ \sum_{i=1}^{n} S(\hat{F}_{x^{(i)}}, (z^{(i)}, \delta^{(i)})), \]
where $\hat{F}_{x^{(i)}}$ is an estimation of the true CDF $F_{x^{(i)}}$ of random variable $T_{x^{(i)}}$ for $x^{(i)} \in X$.  This formulation assumes that each $x^{(i)} \in X$ has an underlying random variable $T_{x^(i)}$ for event times, and our task is to find a good estimation $\hat{F}_{x^{(i)}}$ for each $x^{(i)}$.

In this paper, we investigate the extensions of the scoring rules for survival analysis.  In Sec.~\ref{sec:portnoy}, we consider quantile regression and survival analysis based on quantile regression.  In Secs.~\ref{sec:logarithmic}--\ref{sec:RPS}, we consider distribution regression and survival analysis based on distribution regression.

\subsection{Extension of Pinball Loss}\label{sec:portnoy}

We first review quantile regression~\cite{KB78,KH01}.  Let $Y$ be a real-valued random variable and $F_{Y}$ be its CDF.  In quantile regression, we estimate the $\tau$-th quantile of $Y$, which can be written as
\[ F^{-1}_{Y}(\tau) = \inf \{ y \mid F_{Y}(y) \geq \tau \}. \]
The {\em pinball loss}~\cite{KB78}, which is also known as the {\em check function}, is a widely used scoring rule.  The pinball loss for an estimation $\hat{F}_{Y}$ of $F_{Y}$ and a quantile level $\tau \in [0,1]$ is defined as
\begin{eqnarray}
\lefteqn{S_{\rm Pinball}(\hat{F}_{Y}, y; \tau)} \nonumber \\
& = & \rho_{\tau}(\hat{F}^{-1}_{Y}(\tau), y) \nonumber \\
& = &
\begin{cases}
(1 - \tau)(\hat{F}^{-1}_{Y}(\tau) - y) & \mbox{if} \ \hat{F}^{-1}_{Y}(\tau) \geq y, \\
\tau (y - \hat{F}^{-1}_{Y}(\tau)) & \mbox{if} \ \hat{F}^{-1}_{Y}(\tau) < y.
\end{cases}
\label{eq:pinball}
\end{eqnarray}
Note that the pinball loss with $\tau=0.5$ is equivalent to the mean absolute error (MAE), and it can be used to estimate the {\em median} (i.e., $0.5$-th quantile) of $Y$.  This means that the pinball loss is a generalization of MAE for any quantile level $\tau \in [0,1]$.  Note also that we include the quantile level $\tau$ in the notation $S_{\rm Pinball}(\hat{F}^{-1}_{Y}, y; \tau)$ to clarify that this scoring rule receives $\tau$ as an input.

It is known that the pinball loss is strictly proper (see e.g.,~\cite{GR07}), which means that we have
\[ \mathop{\mathbb{E}}_{y \sim Y} [ S_{\rm Pinball}(\hat{F}_{Y}, y; \tau) ] \geq \mathop{\mathbb{E}}_{y \sim Y} [ S_{\rm Pinball}(F_{Y}, y; \tau) ], \]
and the equality holds only when $\hat{F}^{-1}_{Y}(\tau) = F^{-1}_{Y}(\tau)$ by Definition~\ref{def:original_strict_proper}.  Therefore, it is standard to use the pinball loss both for a loss function and an evaluation metric in quantile regression.

{\em Portnoy's estimator}~\yrcite{Por03} is an extension of the pinball loss for quantile regression-based survival analysis, which is defined as
\begin{eqnarray}
\lefteqn{S_{\rm Portnoy}(\hat{F}, (z,\delta); w, \tau)} \nonumber \\
& = &
\begin{cases}
\rho_{\tau}(\hat{F}^{-1}(\tau), z) & \mbox{if} \ \delta = 1, \\
w \rho_{\tau}(\hat{F}^{-1}(\tau), z) & \\
\ + (1-w) \rho_{\tau}(\hat{F}^{-1}(\tau), z_{\infty}) & \mbox{if} \ \delta = 0,
\end{cases}  \label{eq:Portnoy}
\end{eqnarray}
where $\rho_{\tau}$ is the pinball loss defined in Eq.~(\ref{eq:pinball}), $w$ is a weight parameter to control the balance between two pinball loss terms, and $z_{\infty}$ is any constant such that $z_{\infty} > z_{\max}$.  In Portnoy's estimator, we can set an arbitrary constant $0 \leq w \leq 1$ for the parameter $w$ if $\tau_{c} > \tau$, where $\tau_{c} = \mathrm{Pr}(t \leq c) = F(c)$, but we have to set $w = \mathrm{Pr}(F(c) < F(t) \leq \tau | t > c) = (\tau - \tau_{c})/(1 - \tau_{c})$ otherwise (i.e., $\tau_{c} \leq \tau$).  Since we do not know the true value $\tau_{c} = F(c)$, we have to resolve this problem to use this estimator.  
Before showing how to resolve this problem, we prove that this estimator is proper under the condition that $w$ is correct.  Note that this is the first result for the quantile regression-based survival analysis in terms of the theory of scoring rules.

\begin{theorem}\label{theorem:portnoy}
Portnoy's estimator is proper under the condition that $w$ is correct.
\end{theorem}

\begin{proof}
We give a proof in Appendix~\ref{sec:proof_Portnoy}.
\end{proof}

This theorem means that the crucial part of Portnoy's estimator is to set an appropriate value for $w$, and this theorem ensures that we can recover the true probability distribution $F$ by minimizing Eq.~(\ref{eq:Portnoy}) if $w$ is correct.

Now, we emphasize that we cannot avoid the dependence on unknown parameters such as $F(c)$ in the definition of {\em any} of the scoring rules for survival analysis due to the discretization of $\hat{F}$.  In the case of Portnoy's estimator, even if we know the true value $F^{-1}(\tau_{i})$ for all $\{ \tau_{i} \}_{i=0}^{B}$, we cannot compute $F(c)$ because $c$ is not always contained in $\{ F^{-1}(\tau_{i}) \}_{i=0}^{B}$.  The best we can do is to find quantile levels $\tau_{i}$ and $\tau_{i+1}$ such that $F^{-1}(\tau_{i}) < c \leq F^{-1}(\tau_{i+1})$ by using the assumption that $F$ is a monotonically increasing function.  This means that $F(c)$ is between $\tau_{i}$ and $\tau_{i+1}$.  Even if we could find such $\tau_{i}$ and $\tau_{i+1}$, we would not be able to calculate some important probabilities such as $\mathrm{Pr}(c < t \leq F^{-1}(\tau_{i+1})) = \tau_{i+1} - F(c)$.   Therefore, we usually mitigate this problem by using a large $B$, which enables us to assume, for example, $F^{-1}(\tau_{i+1}) - F^{-1}(\tau_{i}) \approx 0$ for all $i$.

Even if we use a large $B$ to assume that we can find the quantile level $\tau'_{c}$ such that $c \approx F^{-1}(\tau'_{c})$ for any $c$, the problem that we do not know the true $F^{-1}$ remains.  One of the approaches to tackling this problem is the grid search algorithm~\cite{Por03,NBP06}.  In this algorithm, we use a sufficiently large $B$, and we estimate $\hat{F}^{-1}(\tau_{i})$ of $F^{-1}(\tau_{i})$ in the increasing order of $i=0,1,\ldots,B$.   Suppose that we have estimated $\{ \hat{F}^{-1}(\tau_{i}) \}_{i=0}^{j-1}$ and we are going to estimate $\hat{F}^{-1}(\tau_{j})$.  The key idea of this algorithm is that we can find $\tau'_{c} \in \{ \tau_{i} \}_{i=0}^{j-1}$ such that $c \approx \hat{F}^{-1}(\tau'_{c})$ if $\tau_{c} = F(c) < \tau_{j}$.   If we can find such $\tau'_{c}$, we estimate $w$ by using $\tau'_{c} \approx \tau_{c}$.  If we cannot find such $\tau'_{c}$, this algorithm assumes that $\tau_{c} > \tau_{j}$ and we use an arbitrary constant $0 \leq w \leq 1$.  Portnoy~\yrcite{Por03} discusses that this algorithm is analogous to the Kaplan-Meier estimator~\cite{KM58}, and their theoretical analysis~\cite{Por03,NBP06} proves that the estimation model combining Portnoy's estimator, linear regression, and the grid search algorithm can recover the true probability distribution $F$ if there is a sufficient number of data points.

As for another approach, Pearce et al.~\yrcite{PJJZ22} propose the CQRNN algorithm, which we call an {\em iterative reweighting (IR)} algorithm.  Unlike the grid search algorithm, this algorithm estimates $\{ \hat{F}^{-1}(\tau_{i}) \}_{i=0}^{B}$ simultaneously by using a neural network.  This algorithm starts with an arbitrary initial estimation $\hat{F}$, and it estimates $\hat{w}$ of the true $w$ by using $\hat{F}$.  Then, it updates $\hat{F}$ by using $\hat{w}$, and it repeats this iterative procedure of estimating $\hat{F}$ and $\hat{w}$ until these values converge.  This IR algorithm is similar to the expectation-maximization (EM) algorithm, and the relationship between this algorithm and the EM algorithm is discussed in~\cite{PJJZ22}.  Note that this IR algorithm can be implemented for ``free'' according to~\cite{PJJZ22}, which means that we can implement it easily in the computation of the loss function of a neural network training algorithm, and we do not need to construct two separate neural network models for estimating $\hat{F}$ and $\hat{w}$.  The experimental evaluation in~\cite{PJJZ22} shows that the IR algorithm performs the best among the quantile regression-based survival analysis models.

\subsection{Extension of Logarithmic Score}\label{sec:logarithmic}

While we estimate $\{ \hat{F}^{-1}_{Y}(\tau_{i}) \}_{i=0}^{B}$ in quantile regression, we consider distribution regression, in which we estimate $\{ \hat{F}_{Y}(\zeta_{i}) \}_{i=0}^{B}$.  For distribution regression, the logarithmic score~\cite{Good52} is known as a strictly proper scoring rule, and it is defined as
\begin{eqnarray}
\lefteqn{S_{\rm Log}(\hat{F}_{Y}, y; \{ \zeta_{i} \}_{i=0}^{B})} \nonumber \\
& = & - \sum_{i=0}^{B-1} \mathbbm{1}(\zeta_{i} < y \leq \zeta_{i+1}) \log(\hat{F}_{Y}(\zeta_{i+1}) - \hat{F}_{Y}(\zeta_{i})) \nonumber \\
& = & - \sum_{i=0}^{B-1} \mathbbm{1}(\zeta_{i} < y \leq \zeta_{i+1}) \log \hat{f}_{i},  \label{eq:logarithmic}
\end{eqnarray}
where $\hat{f}_{i} = \hat{F}_{Y}(\zeta_{i+1}) - \hat{F}_{Y}(\zeta_{i})$ for $i=0,1,\ldots,B-1$.

We extend this logarithmic score for distribution regression-based survival analysis as
\begin{eqnarray}
\lefteqn{S_{\rm Cen-log}(\hat{F}, (z, \delta); \{ w_{i} \}_{i=0}^{B-1}, \{ \zeta_{i} \}_{i=0}^{B})} \nonumber \\
& = & - \sum_{i=0}^{B-1} \mathbbm{1}(\zeta_{i} < z \leq \zeta_{i+1}) g(i,\delta, w_{i}),  \label{eq:logarithmic_censored}
\end{eqnarray}
where
\begin{eqnarray*}
\lefteqn{g(i,\delta, w_{i})} \\
& = & 
\begin{cases}
\log \hat{f}_{i} & \ \mbox{if} \ \delta=1, \\
w_{i} \log \hat{f}_{i} + (1-w_{i}) \log (1 - \hat{F}(\zeta_{i+1})) & \ \mbox{if} \ \delta=0,
\end{cases}
\end{eqnarray*}
$\hat{f}_{i} = \hat{F}(\zeta_{i+1}) - \hat{F}(\zeta_{i})$, and $w_{i} = \mathrm{Pr}(c < t \leq \zeta_{i+1} | t > c) = (F(\zeta_{i+1}) - F(c))/(1 - F(c))$.  Note that this scoring rule is equivalent to Eq.~(\ref{eq:logarithmic}) if $\delta=1$.  Similar to Portnoy's estimator, we cannot set the parameter $w_{i}$ of this scoring rule because we do not know $F(\zeta_{i+1})$ and $F(c)$.

Even though we do not know the correct $\{ w_{i} \}_{i=0}^{B-1}$, we prove that this scoring rule is proper if the set of parameters $\{ w_{i} \}_{i=0}^{B-1}$ is correct.

\begin{theorem}\label{theorem:logarithmic}
$S_{\rm Cen-Log}(\hat{F}, (z, \delta); \{ w_{i} \}_{i=0}^{B-1}, \{ \zeta_{i} \}_{i=0}^{B})$ is a proper scoring rule under the condition that $w_{i}$ is correct for all $i$.
\end{theorem}

\begin{proof}
We give a proof in Appendix~\ref{sec:variant_logarithmic}.
\end{proof}

Similar to Portnoy's estimator, we can use both the grid-search algorithm and the IR algorithm to estimate $\{ w_{i} \}_{i=0}^{B-1}$.

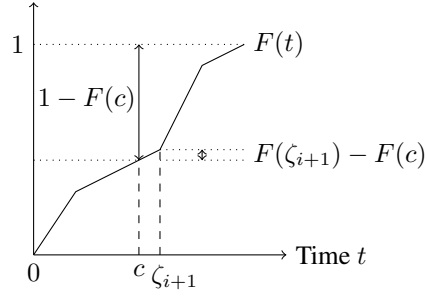
\begin{figure}[t]
\centering
\begin{tikzpicture}[scale=2.8]
\draw (0,0)--(0.2,0.3)--(0.6,0.5)--(0.8,0.9)--(1.0,1.0);
\draw[dotted] (0,1)--(1,1);
\draw[dotted] (0,0.45)--(0.5,0.45);
\draw[<->] (0.5,0.45) -- (0.5,1);
\node at (0.53,0.74) [left] {$1-F(c)$};
\draw[dotted] (0.6,0.45)--(1,0.45);
\draw[dotted] (0.6,0.5)--(1,0.5);
\draw[<->] (0.8,0.45) -- (0.8,0.5);
\node at (1.0,0.475) [right] {$F(\zeta_{i+1})-F(c)$};
\draw[dashed] (0.5,0)--(0.5,0.45);
\draw[dashed] (0.6,0)--(0.6,0.5);
\node at (0,0) [below] {$0$}; 
\node at (0.5,0) [below] {$c$};
\node at (0.67,0) [below] {$\zeta_{i+1}$};
\draw [->] (0,0)--(0,1.2);
\draw [->] (0,0)--(1.2,0);
\node at (0,1.0) [left] {$1$};
\node at (1.0,1.0) [right] {$F(t)$};
\node at (1.2,0) [right] {Time $t$};
\end{tikzpicture}
\caption{Illustration of computation of weight $w_{i}=(F(\zeta_{i+1})-F(c))/(1-F(c))$ for scoring rule $S_{\rm Cen-log}$.}
\label{fig:weight_w_logarithm}
\end{figure}

In addition, we show another simpler approach by assuming that $w_{i} \approx 0$ for all $i$ if $B$ is large.  If $B$ is large, $1-F(c)$ is usually much larger than $F(\zeta_{i+1}) - F(c)$ (see Fig.~\ref{fig:weight_w_logarithm}), and hence, we have $w_{i} = (F(\zeta_{i+1}) - F(c))/(1-F(c)) \approx 0$.  Therefore, we can obtain a simpler variant of $S_{\rm Cen-log}$ by setting $w_{i} = 0$ for all $i$:
\begin{eqnarray}
\lefteqn{S_{\rm Cen-log-simple}(\hat{F}, (z, \delta); \{ \zeta_{i} \}_{i=0}^{B})} \nonumber \\
& = & -  \sum_{i=0}^{B-1} \mathbbm{1}(\zeta_{i} < z \leq \zeta_{i+1}) g(i,\delta, 0) \label{eq:logarithmic_simple} \\ 
& = & - \delta \sum_{i=0}^{B-1} \mathbbm{1}(\zeta_{i} < z \leq \zeta_{i+1}) \log \hat{f}_{i} \nonumber \\
& & - (1-\delta) \sum_{i=0}^{B-1} \mathbbm{1}(\zeta_{i} < z \leq \zeta_{i+1}) \log (1 - \hat{F}(\zeta_{i+1})) \nonumber. 
\end{eqnarray}
Furthermore, by increasing $B$ to infinity (i.e., $B \rightarrow \infty$), we obtain the continuous version of this scoring rule:
\begin{eqnarray}
\lefteqn{S_{\rm Cen-cont-log}(\hat{F}, (z, \delta))} \nonumber \\
& = & - \delta \log \frac{d\hat{F}}{dt}(z) - (1-\delta) \log (1 - \hat{F}(z)),  \label{eq:logarithmic_infty}
\end{eqnarray}
which is equal to the extension of the logarithmic score that is proven to be strictly proper in~\cite{RHSS22}.

\paragraph{Remarks.}
This simplification clarifies that the proof in~\cite{RHSS22} implicitly assumes that $B$ is sufficiently large.  This means that we should set $B$ large enough in practice.  Moreover, strictly speaking, the relation $w_{i} = (F(\zeta_{i+1}) - F(c))/(1-F(c)) \approx 0$ may not hold if $1 - F(c) \approx 0$.  Therefore, we recommend $S_{\rm Cen-log}$ (Eq.~(\ref{eq:logarithmic_censored})) rather than $S_{\rm Cen-log-simple}$ (Eq.~(\ref{eq:logarithmic_simple})) and $S_{\rm Cen-cont-log}$ (Eq.~(\ref{eq:logarithmic_infty})) whenever possible.

\subsection{Extension of Brier Score}\label{sec:Brier}

In distribution regression, the Brier score~\cite{Brier50} is also known as a strictly proper scoring rule, which is defined as
\begin{eqnarray}
\lefteqn{S_{\rm Brier}(\hat{F}_{Y}, y; \{ \zeta_{i} \}_{i=0}^{B})} \nonumber \\
& = & \sum_{i=0}^{B-1} (\mathbbm{1}(\zeta_{i} < y \leq \zeta_{i+1}) - \hat{f}_{i})^2,  \label{eq:Brier}
\end{eqnarray}
where $\hat{f}_{i} = \hat{F}_{Y}(\zeta_{i+1}) - \hat{F}_{Y}(\zeta_{i})$ for $i=0,1,\ldots,B-1$.

We extend this Brier score for distribution regression-based survival analysis as
\begin{eqnarray}
\lefteqn{S_{\rm Cen-Brier}(\hat{F}, (z, \delta); \{ w_{i} \}_{i=0}^{B-1}, \{ \zeta_{i} \}_{i=0}^{B})} \nonumber \\
& = & \sum_{i=0}^{B-1} \left( w_{i} (1-\hat{f}_{i})^{2} + (1 - w_{i}) \hat{f}_{i}^{2} \right),  \label{eq:cen-Brier}
\end{eqnarray}
where
\[
w_{i} =
\begin{cases}
0 & \mbox{if} \ \delta=1 \ \mbox{and} \ \zeta_{i+1} < z=t, \\
1 & \mbox{if} \ \delta=1 \ \mbox{and} \ \zeta_{i} < z = t \leq \zeta_{i+1}, \\
0 & \mbox{if} \ z \leq \zeta_{i}, \\
 \frac{F(\zeta_{i+1}) - F(c)}{1 - F(c)} & \mbox{if} \ \delta=0 \ \mbox{and} \ \zeta_{i} < z = c \leq \zeta_{i+1}, \\
 \frac{F(\zeta_{i+1}) - F(\zeta_{i})}{1-F(c)} & \mbox{if} \ \delta=0 \ \mbox{and} \ \zeta_{i+1} < z=c.
\end{cases}
\]
If $\delta=1$, it is easy to see that Eq.~(\ref{eq:cen-Brier}) is equivalent to Eq.~(\ref{eq:Brier}).

We prove that this scoring rule is proper if the set of parameters $\{ w_{i} \}_{i=0}^{B-1}$ is correct.

\begin{theorem}\label{theorem:Brier}
$S_{\rm Cen-Brier}(\hat{F}, (z, \delta); \{ w_{i} \}_{i=0}^{B-1}, \{ \zeta_{i} \}_{i=0}^{B})$ is a proper scoring rule under the condition that $w_{i}$ is correct for all $i$.
\end{theorem}

\begin{proof}
We give a proof in Appendix~\ref{sec:variant_Brier}.
\end{proof}

We can use the IR algorithm to estimate $w_{i}$.  However, unlike Portnoy's estimator and the extension of the logarithmic score, we cannot use the grid-search algorithm in this extension of the Brier score because we need to estimate $w_{i}$ for all $i=0,1,\ldots,B-1$.

Note that each $w_{i}$ in this scoring rule is close to zero if $B$ is large and $\delta=0$.  However, since $w_{i}$s are designed to satisfy $\sum_{i} w_{i} = 1$, we cannot use the approximation $w_{i} \approx 0$ for this scoring rule.

\subsection{Extension of Ranked Probability Score}\label{sec:RPS}

The ranked probability score (RPS) is also known as a strictly proper scoring rule in distribution regression (see, e.g.,~\cite{GR07}).  It is defined as
\begin{equation*}
S_{\rm RPS}(\hat{F}_{Y},y) = \sum_{i=1}^{B-1} S_{\rm Binary-Brier}(\hat{F}_{Y}, y; \zeta_{i}),
\end{equation*}
where $S_{\rm Binary-Brier}$ is the binary version of $S_{\rm Brier}$ (Eq.~(\ref{eq:Brier})) with a single threshold $\zeta$:
\begin{equation*}
S_{\rm Binary-Brier}(\hat{F}_{Y}, y; \zeta) = (\mathbbm{1}(y \leq \zeta) - \hat{F}_{Y}(\zeta))^2.
\end{equation*}

We extend this scoring rule for survival analysis:
\begin{eqnarray}
\lefteqn{S_{\rm Cen-RPS}(\hat{F}, (z, \delta);  \{ w_{i} \}_{i=1}^{B-1}, \{ \zeta_{i} \}_{i=1}^{B-1})} \nonumber \\
& = & \sum_{i=1}^{B-1} S_{\rm Cen-Binary-Brier}(\hat{F}, (z, \delta); w_{i}, \zeta_{i}),  \label{eq:censored_RPS}
\end{eqnarray}
where $S_{\rm Cen-Binary-Brier}$ is the binary version of $S_{\rm Cen-Brier}$ (Eq.~(\ref{eq:cen-Brier})) with a single threshold $\zeta$:
\begin{eqnarray*}
\lefteqn{S_{\rm Cen-Binary-Brier}(\hat{F}, (z, \delta); w, \zeta)} \\
& = & 
\begin{cases}
 \hat{F}(\zeta)^2 & \mbox{if} \ z > \zeta, \\
 (1- \hat{F}(\zeta))^2 & \mbox{if} \ \delta=1 \ \mbox{and} \ z = t \leq \zeta, \\
 w (1- \hat{F}(\zeta))^2 \\
 \ + (1-w) \hat{F}(\zeta)^2 & \mbox{if} \ \delta=0 \ \mbox{and} \ z = c \leq \zeta,
\end{cases}
\end{eqnarray*}
where $w = (F(\zeta) - F(c))/(1 - F(c))$.

Since this scoring rule is just the sum of the binary version of Brier scores for survival analysis, it is straightforward to prove this theorem.

\begin{theorem}
$S_{\rm Cen-RPS}(\hat{F}, (z, \delta); \{ w_{i} \}_{i=1}^{B-1}, \{ \zeta_{i} \}_{i=1}^{B-1})$ is a proper scoring rule under the condition that $w_{i}$ is correct for all $i$.
\end{theorem}

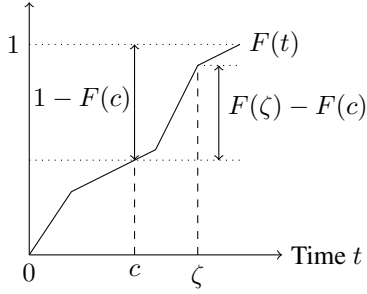
\begin{figure}[t]
\centering
\begin{tikzpicture}[scale=2.8]
\draw (0,0)--(0.2,0.3)--(0.6,0.5)--(0.8,0.9)--(1.0,1.0);
\draw[dotted] (0,1)--(1,1);
\draw[dotted] (0,0.45)--(1,0.45);
\draw[<->] (0.5,0.45) -- (0.5,1);
\node at (0.53,0.74) [left] {$1-F(c)$};
%\draw[dotted] (0.8,0.45)--(1,0.45);
\draw[dotted] (0.8,0.9)--(1,0.9);
\draw[<->] (0.9,0.45) -- (0.9,0.9);
\node at (0.9,0.7) [right] {$F(\zeta)-F(c)$};
\draw[dashed] (0.5,0)--(0.5,0.45);
\draw[dashed] (0.8,0)--(0.8,0.9);
\node at (0,0) [below] {$0$}; 
\node at (0.5,0) [below] {$c$};
\node at (0.8,0) [below] {$\zeta$};
\draw [->] (0,0)--(0,1.2);
\draw [->] (0,0)--(1.2,0);
\node at (0,1.0) [left] {$1$};
\node at (1.0,1.0) [right] {$F(t)$};
\node at (1.2,0) [right] {Time $t$};
\end{tikzpicture}
\caption{Illustration of computations of weight $w_{i}=(F(\zeta)-F(c))/(1-F(c))$ for scoring rule $S_{\rm Cen-Binary-Brier}$.}
\label{fig:weight_w_binary_Brier}
\end{figure}

Note that the scoring rule $S_{\rm Cen-Binary-Brier}$ is analogous to Portnoy's estimator.  The scoring rule $S_{\rm Cen-Binary-Brier}$ is designed to estimate $\hat{F}(\zeta)$, where $\zeta$ is an input, and we use $F(c)$ and $\zeta$ to set $w$, whereas Portnoy's estimator is designed to estimate $\hat{F}^{-1}(\tau)$, where $\tau$ is an input, and we use $F(c)$ and $\tau$ to set $w$.  As these two scoring rules are similar, we can use both the grid-search algorithm and the IR algorithm for $S_{\rm Cen-RPS}$.

Unlike $S_{\rm Cen-log}$ defined in Eq.~(\ref{eq:logarithmic_censored}), the parameter $w$ of the scoring rule $S_{\rm Cen-Binary-Brier}$ is usually not close to zero because $\zeta$ and $c$ are usually not close to each other as shown in Fig.~\ref{fig:weight_w_binary_Brier}.  We note that the parameter $w$ of Portnoy's estimator is also not close to zero for a similar reason.

\section{Evaluation Metrics for Survival Analysis}\label{sec:metric}

While we have discussed the extensions of the scoring rules as loss functions, we should use strictly proper scoring rules also for evaluation metrics.  However, among the extensions of the scoring rules for survival analysis, we can use only $S_{\rm Cen-log-simple}$ (Eq.~(\ref{eq:logarithmic_simple})) as an evaluation metric because the other scoring rules depend on the parameter $w$ or $\{ w_{i} \}_{i=0}^{B-1}$.   Note that we can use $S_{\rm Cen-log-simple}$ only when $B$ is sufficiently large.   In Appendix~\ref{sec:experiments_appendix}, we conducted experiments on choosing an appropriate $B$, and the results suggested using $B > 16$.

While we can use $S_{\rm Cen-log-simple}$ as a discrimination metric for survival analysis, we note that there is a calibration metric, D-calibration~\cite{HHDG20}, for survival analysis.  D-calibration is widely used in survival analysis, but we propose another calibration metric, {\em KM-calibration}.  Let $\kappa(t)$ be the survival function estimated by the Kaplan-Meier estimator~\cite{KM58}.  This function $\kappa(t)$ represents the survival rate (i.e., the probability that the event time is less than $t$) over the entire dataset rather than individual feature vector $x$.  By definition, $\kappa(0)=1$ and $\kappa(t)$ is a monotonically decreasing function.  Assuming that $\kappa(t)$ is correct, $\kappa(t) = 1 - \hat{F}_{\rm avg}(t)$ must hold, where $\hat{F}_{\rm avg}(t)$ is the average of $\hat{F}(t)$ over all data points in the test dataset.   Therefore, we define our KM-calibration as the Kullback-Leibler divergence between $\kappa(t)$ and $1 - \hat{F}_{\rm avg}(t)$:
\begin{eqnarray*}
d_{\mathrm{KM-cal}}(\kappa, \hat{F}_{\rm avg})
& = & d_{\mathrm{KL}}(\kappa || 1-\hat{F}_{\rm avg}) \\
& = & \sum_{i=0}^{B-1} (p_{i} \log p_{i} - p_{i} \log q_{i}),
\end{eqnarray*}
where $p_{i} = \kappa(\zeta_{i+1}) - \kappa(\zeta_{i})$, $q_{i} = (1 - \hat{F}_{\rm avg}(\zeta_{i+1})) - (1 - \hat{F}_{\rm avg}(\zeta_{i}))$, and we assume here that $\kappa(\zeta_{B}) = 0$.  This metric is based on the observation that the model's predicted number of events within any time interval should be similar to the observed number~\cite{GHPPR20}.  We note that there is another calibration metric~\cite{CTL20} based on the Kaplan-Meier estimator.  Whereas this calibration metric uses the absolute difference, our KM-calibration uses the Kullback-Leibler divergence.

\section{Experiments}

\begin{figure*}
\centering
  \subfigure[$S_{\rm Cen-log-simple}$ on flchain]{
    \includegraphics[width=0.31\textwidth]{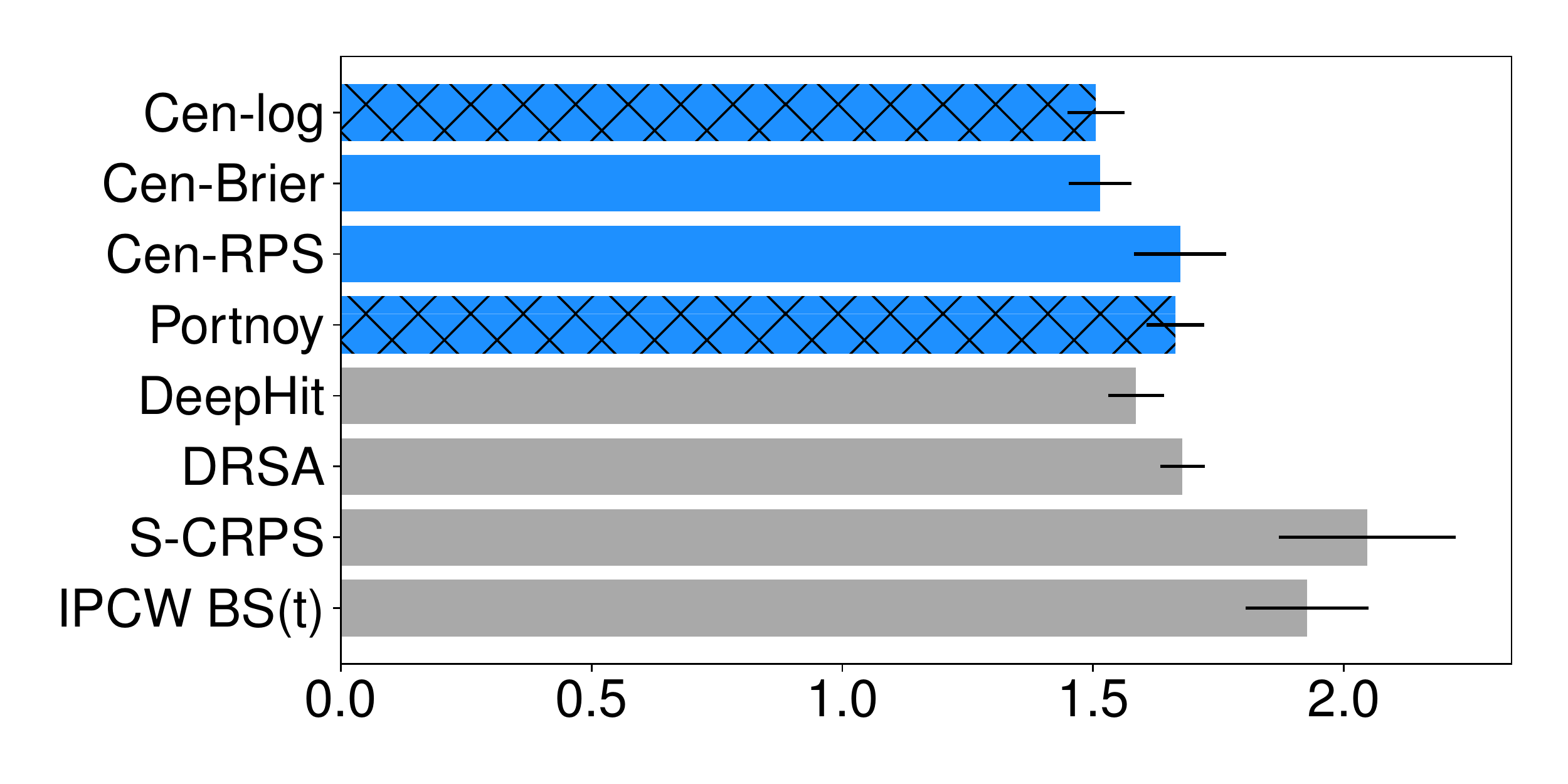}
  }\label{fig:flchain_logarithmic}
  \subfigure[$S_{\rm Cen-log-simple}$ on prostateSurvival]{
    \includegraphics[width=0.31\textwidth]{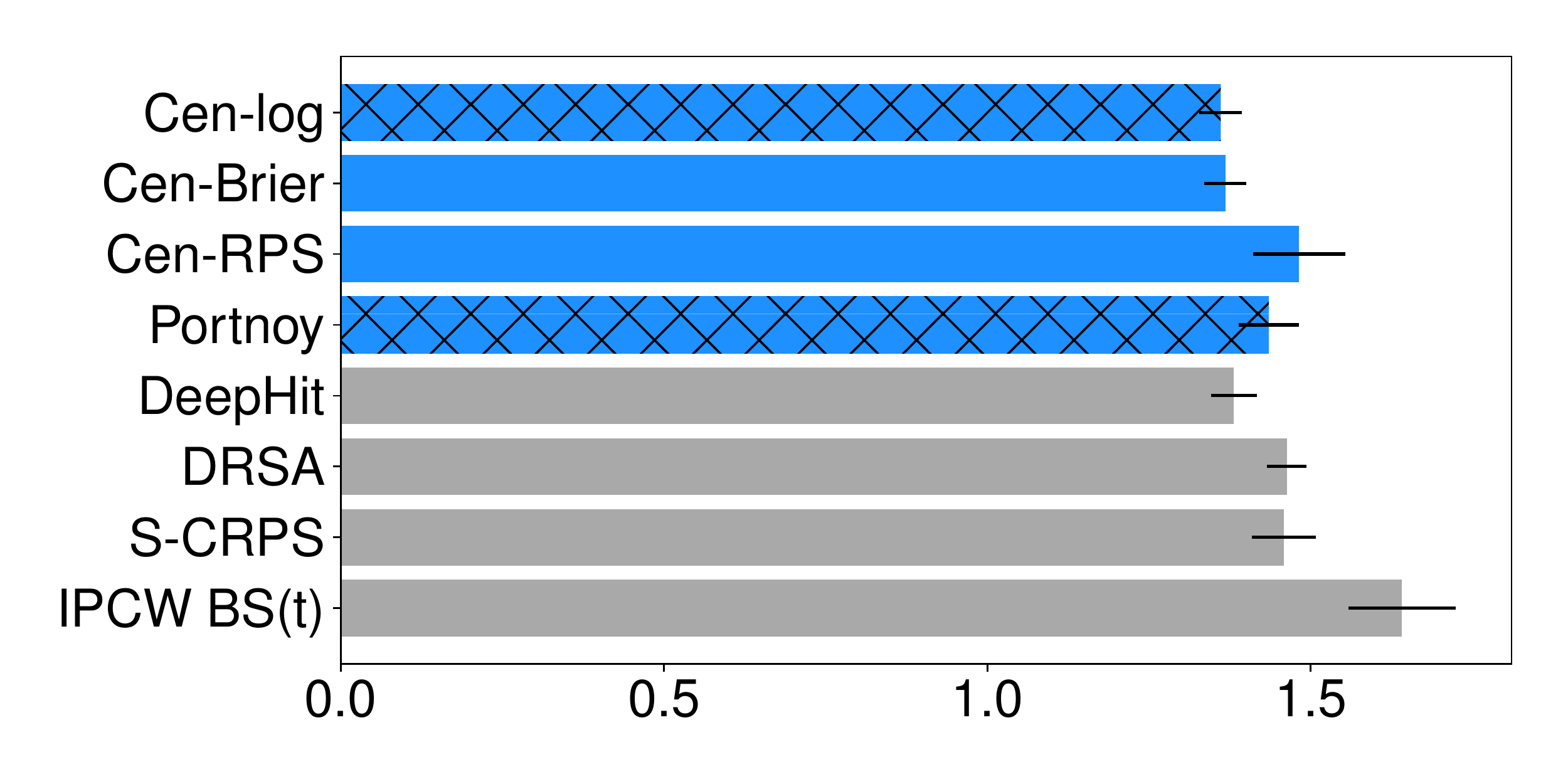}
  }\label{fig:prostateSurvival_logarithmic}
  \subfigure[$S_{\rm Cen-log-simple}$ on support]{
    \includegraphics[width=0.31\textwidth]{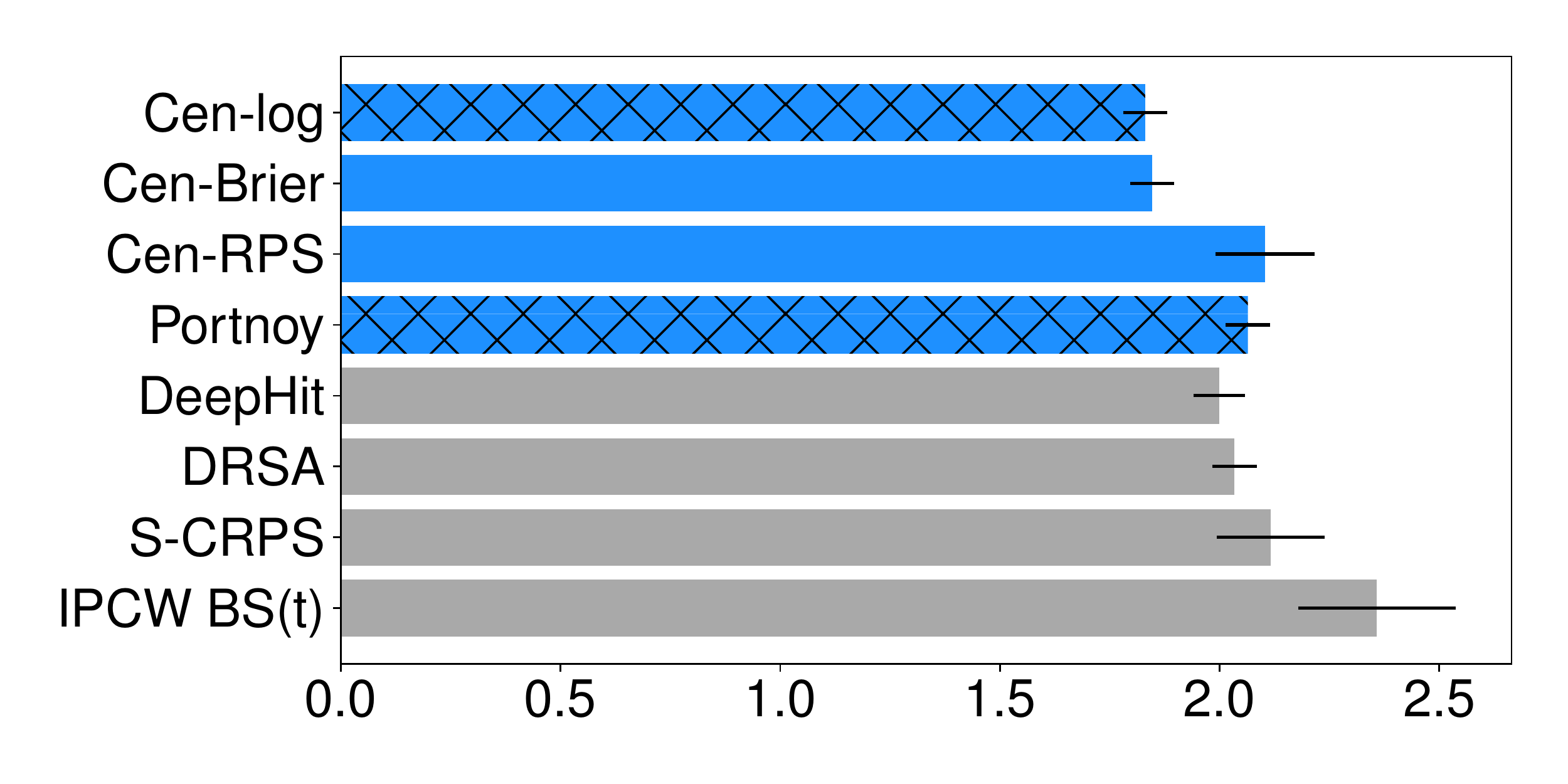}
  }\label{fig:support_logarithmic}
  \subfigure[D-calibration on flchain]{
    \includegraphics[width=0.31\textwidth]{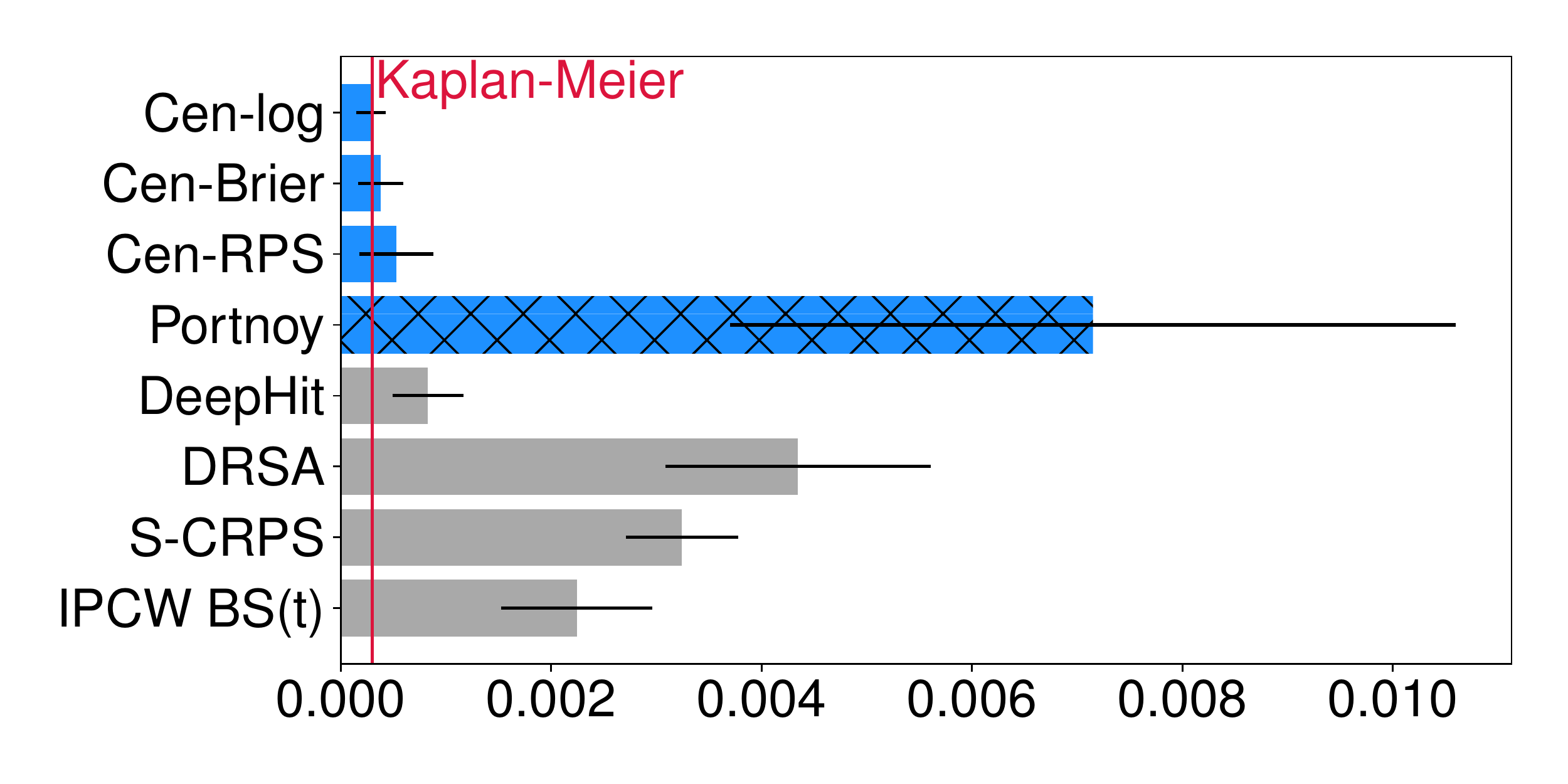}
  }\label{fig:flchain_D}
  \subfigure[D-calibration on prostateSurvival]{
    \includegraphics[width=0.31\textwidth]{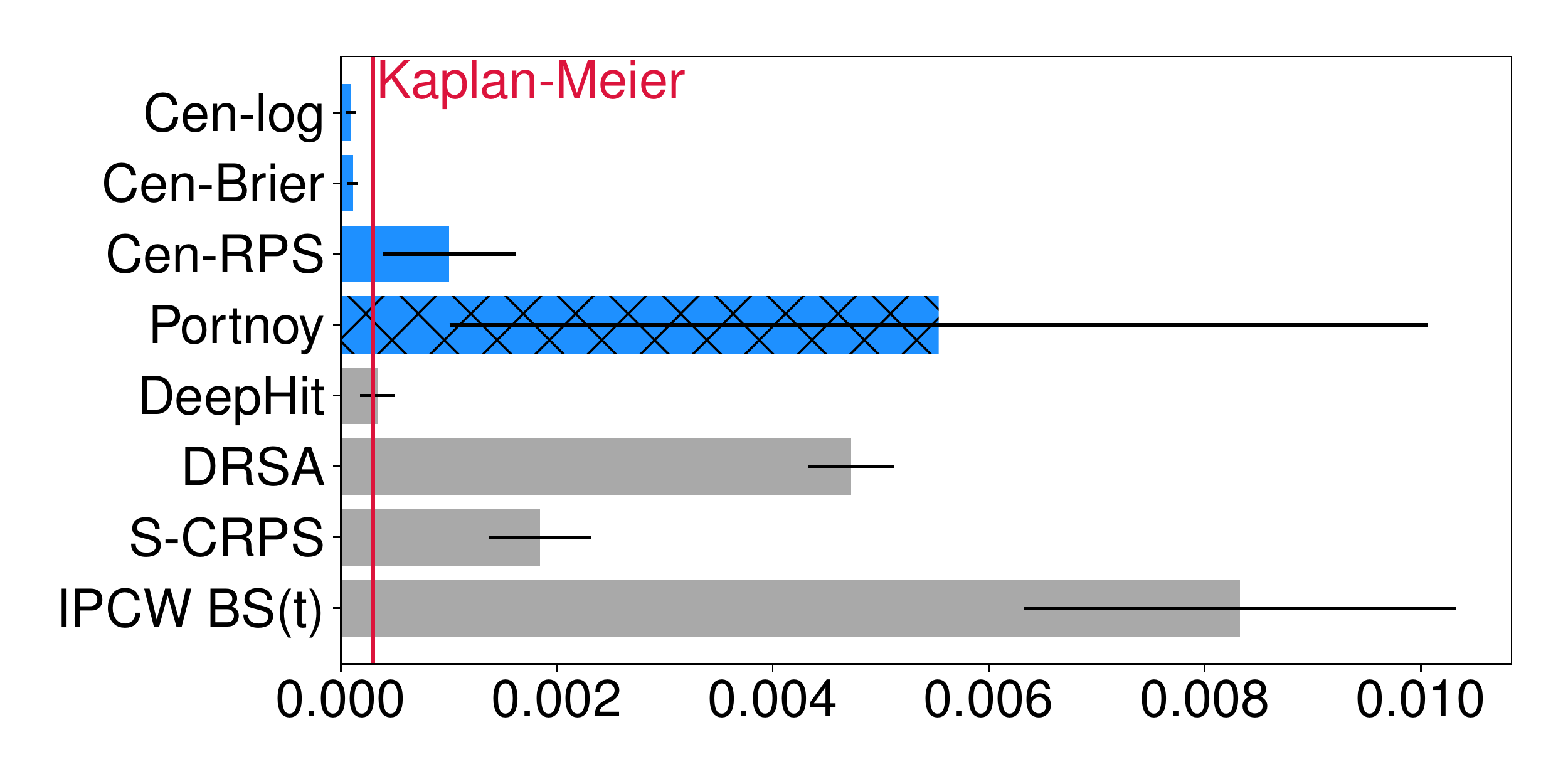}
  }\label{fig:prostateSurvival_D}
  \subfigure[D-calibration on support]{
    \includegraphics[width=0.31\textwidth]{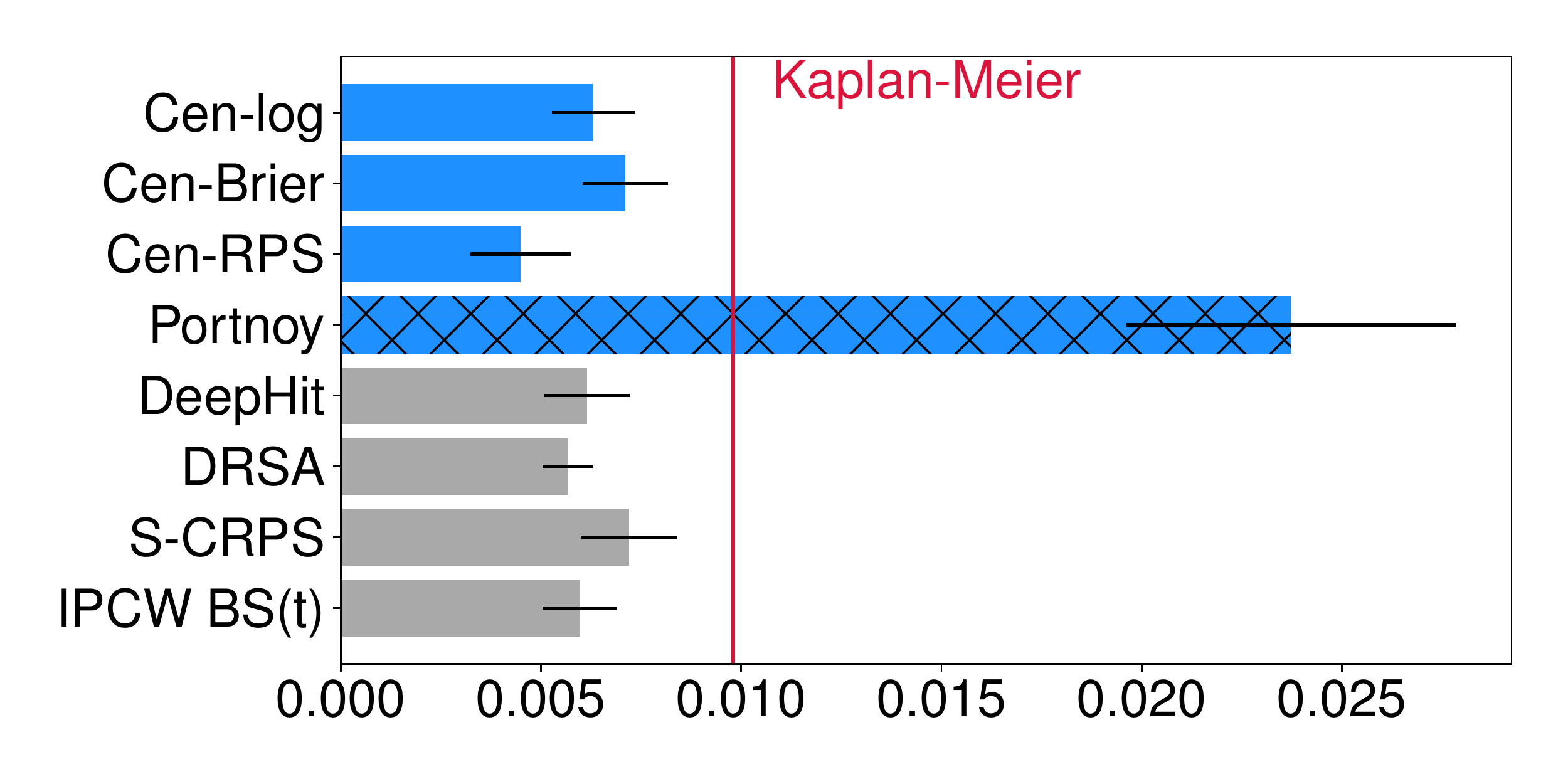}
  }\label{fig:support_D}
  \subfigure[KM-calibration on flchain]{
    \includegraphics[width=0.31\textwidth]{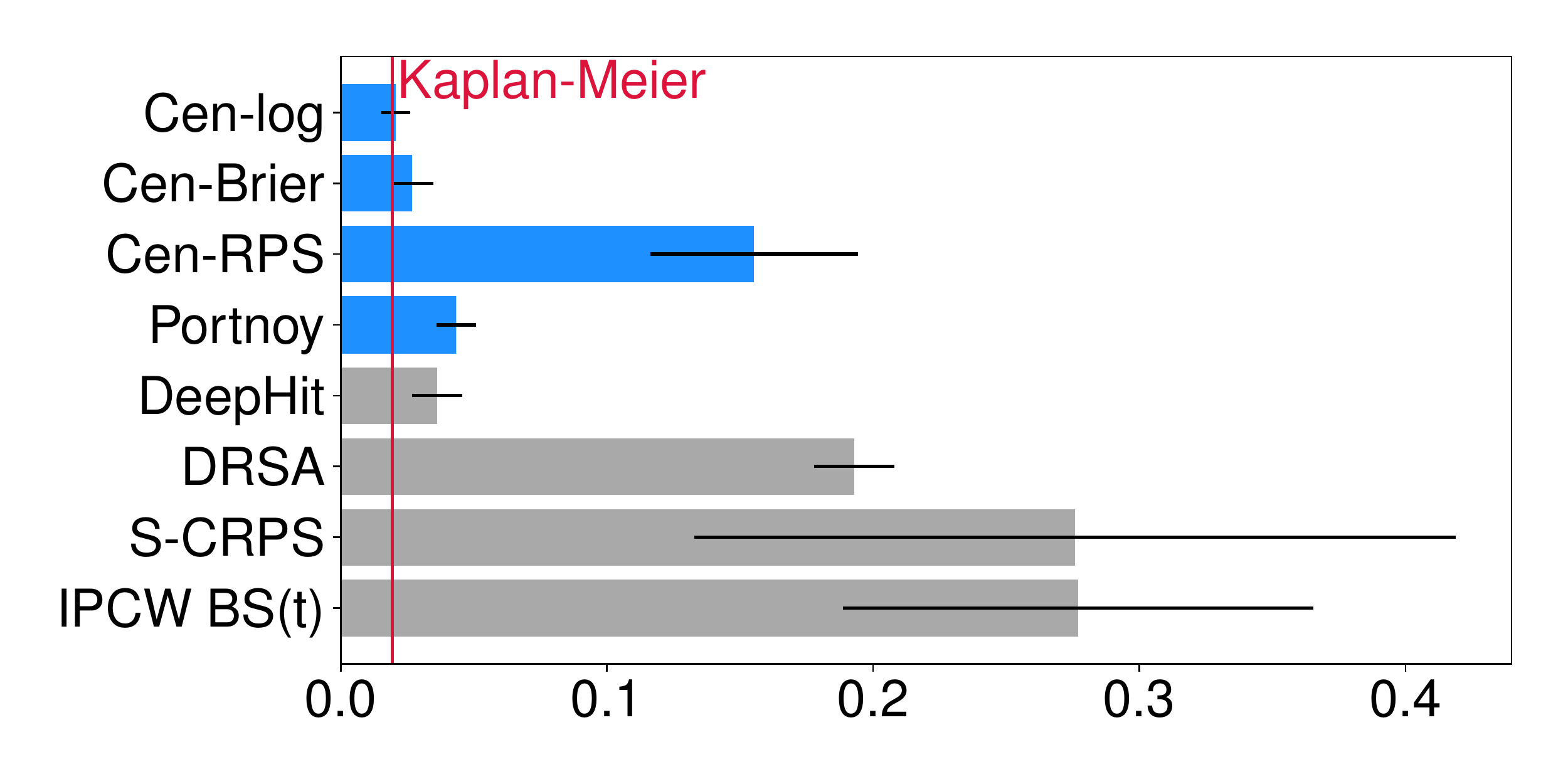}
  }\label{fig:flchain_KM}
  \subfigure[KM-calibration on prostateSurvival]{
    \includegraphics[width=0.31\textwidth]{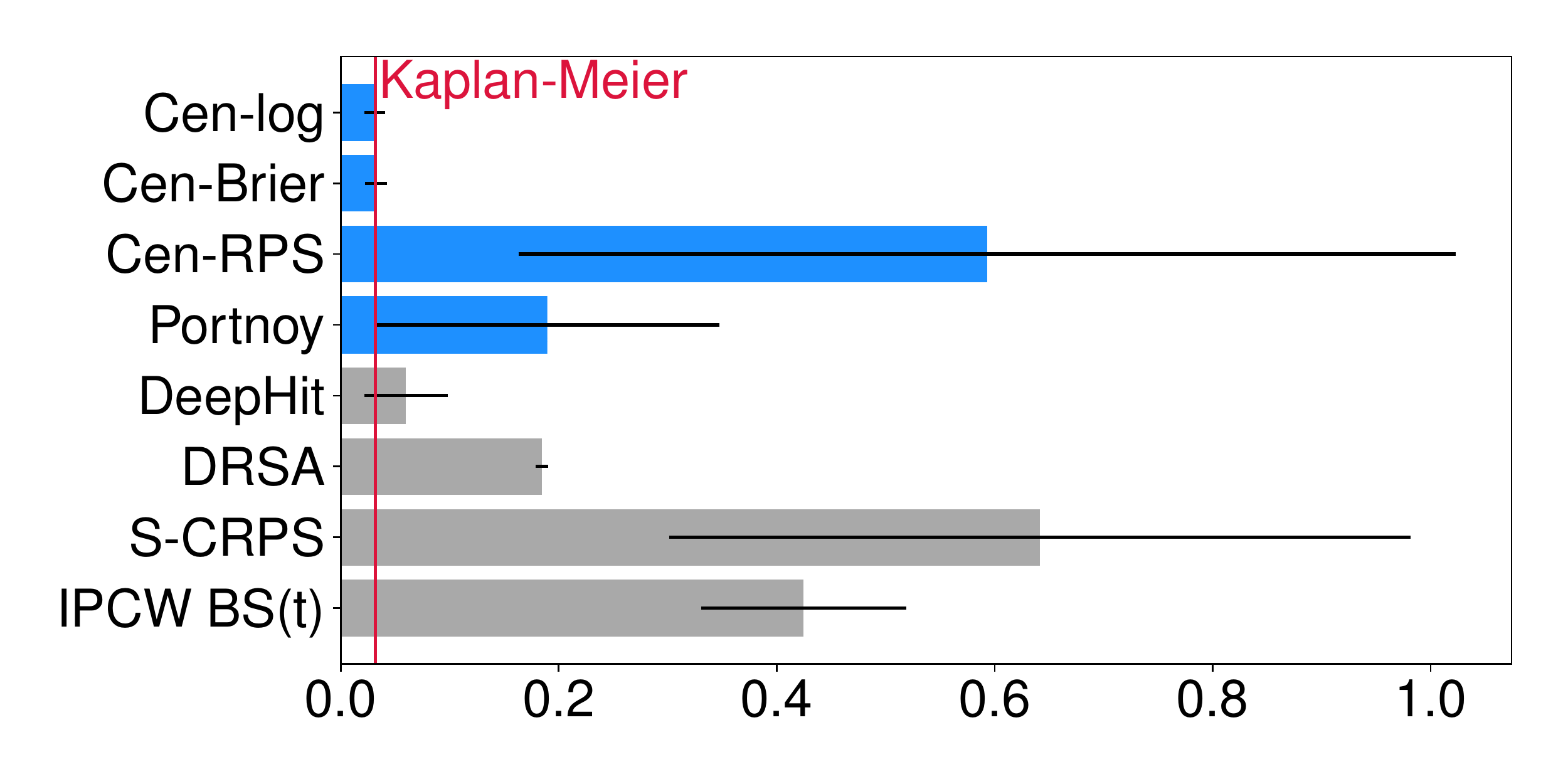}
  }\label{fig:prostateSurvival_KM}
  \subfigure[KM-calibration on support]{
    \includegraphics[width=0.31\textwidth]{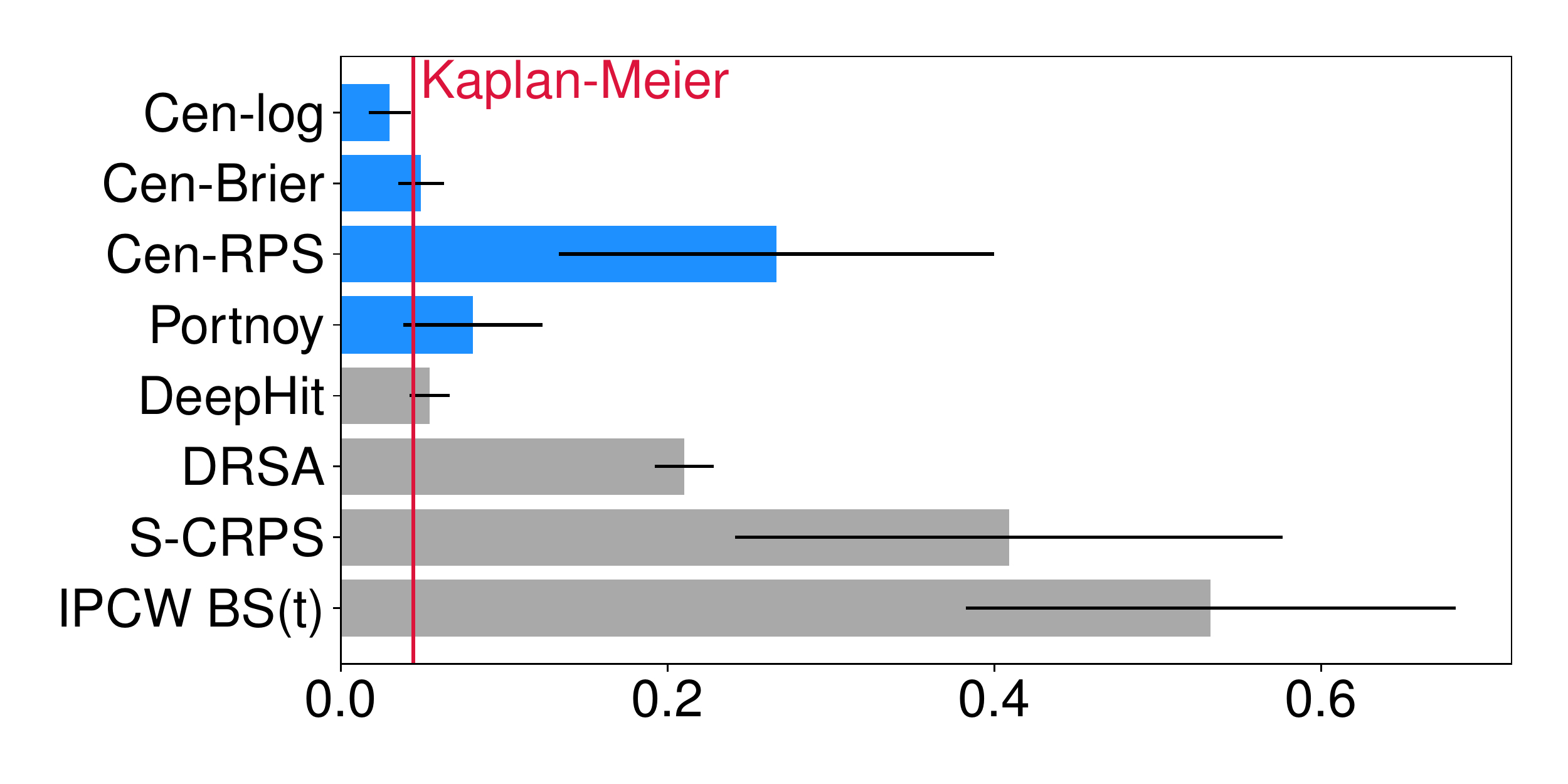}
  }\label{fig:support_KM}
\caption{Prediction performance (lower is better) comparison on three datasets with $S_{\rm Cen-log-simple}$, KM-calibration, and D-calibration.}
\label{fig:comparison}
\end{figure*}

In our experiments, we compared practical prediction performances of various loss functions on real datasets.  We used three datasets for the survival analysis from the packages in R~\cite{R}: the flchain dataset~\cite{DKKL12}, which was obtained from the ``survival'' package and contains 7874 data points (69.9\% of which are censored), the prostateSurvival dataset~\cite{LAMS09}, which was obtained from the ``asaur'' package and contains 14294  data points (71.7\% of which are censored), and the support dataset~\cite{KHLG95}, which was obtained from the ``casebase'' package and contains 9104 data points (31.9\% of which are censored).  For each dataset, we divided the time interval $[0,z_{\max}+\epsilon)$, where $\epsilon = 10^{-3}$, into $B$ equal-length intervals to get the thresholds $\{ \zeta_{i} \}_{i=0}^{B}$ for distribution regression-based survival analysis, and we divided the unit interval $[0,1]$ into $B$ equal-length intervals to get the quantile levels $\{ \tau_{i} \}_{i=0}^{B}$ for quantile regression-based survival analysis.  Unless otherwise stated, we set $B=32$.

All our experiments were conducted on a virtual machine with an Intel Xeon CPU (3.30 GHz) processor without any GPU and 64 GB of memory running Red Hat Enterprise Linux Server 7.6.  We used Python 3.7.4 and PyTorch 1.7.1 for the implementation.

We estimated $\hat{F}(t)$ by combining a multi-layer perceptron (MLP) and the IR algorithm (see Sec.~\ref{sec:portnoy}) to estimate $w$ or $\{ w_{i} \}_{i=0}^{B-1}$.  The MLP consists of three hidden layers containing 128 neurons, and the number of outputs was $B$. The type of activation function after the hidden layer was the rectified linear unit (ReLU), and the activation function at the output node was softmax.  The softmax function is used to satisfy the assumption that $\hat{F}(t)$ is a monotonically increasing continuous function.  In distribution regression-based survival analysis, each output of MLP estimates $\hat{f}_{i} = \hat{F}(\zeta_{i+1}) - \hat{F}(\zeta_{i})$ for $i=0,1,\ldots,B-1$.  By using these outputs $\{ \hat{f}_{i} \}_{i=0}^{B-1}$, we can calculate $\{ \hat{F}(\zeta_{i}) \}_{i=0}^{B}$ and we can represent the function $\hat{F}(t)$ as a piecewise linear function connecting the values $\{ \hat{F}(\zeta_{i}) \}_{i=0}^{B}$.  Since $\hat{f}_{i} > 0$ holds for all $i$, $\hat{F}(t)$ estimated in this way is a monotonically increasing continuous function.  We can estimate $\hat{F}$ for quantile regression-based survival analysis by using a similar way.

For the training of the neural network, we used the Adam optimizer~\cite{KB15} with the learning rate $0.001$, and the other parameters were set to their default values.  We ran training for $300$ epochs for our neural network models.  Our implementation of the scoring rules are available at \texttt{https://github.com/IBM/dqs}.

We compared the prediction performances of various scoring rules (i.e., loss functions), and Fig.~\ref{fig:comparison} shows the results.  In these experiments, we split the data points into training (60\%), validation (20\%), and test (20\%), and each bar shows the mean of the measurements on the test data of five random splits together with the error bar, which represents the standard deviation.  We used $S_{\rm Cen-log-simple}$ (Eq.~(\ref{eq:logarithmic_simple})) as a metric for discrimination performance and D-calibration~\cite{HHDG20} and KM-calibration (see Sec.~\ref{sec:metric}) as calibration metrics, where we used 20 bins of equal length for D-calibration.  For the calibration metrics, we added the mean D-calibration and mean KM-calibration of the Kaplan-Meier estimator~\cite{KM58} as a red line in each graph.  Since the Kaplan-Meier estimator is calibrated in theory, the values of the D-calibration and the KM-calibration of this estimator should be regarded as close to zeros.   In this figure, the four scoring rules Cen-log ($S_{\rm Cen-log}$ defined in Eq.~(\ref{eq:logarithmic_censored})), Cen-Brier ($S_{\rm Cen-Brier}$ defined in Eq.~(\ref{eq:cen-Brier})), Cen-RPS ($S_{\rm Cen-RPS}$ defined in Eq.~(\ref{eq:censored_RPS})), and Portnoy ($S_{\rm Portnoy}$ defined in Eq.~(\ref{eq:Portnoy})) are proved to be conditionally proper in this paper.  Note that Cen-log is similar to the scoring rule (Eq.~(\ref{eq:logarithmic_infty})) that is proved to be strictly proper in~\cite{RHSS22} and Portnoy is proposed in~\cite{Por03}.  This figure also contains the results for other scoring rules in the state-of-the-art models for survival analysis: DeepHit~\cite{LZYS18} with parameter $\alpha=1$, DRSA~\cite{RQZYZQY19} with parameter $\alpha=0.25$, S-CRPS~\cite{ADZJSN19}, and IPCW BS($t$) game model~\cite{HGPWPR21}.  These four scoring rules are not proved to be proper.

\begin{table*}[t]
  \caption{Prediction performances of DeepHit (lower is better) with various $\alpha$ for $B=32$.}
  \label{table:deephit_alpha}
  \centering
  \begin{tabular}{ll|lll}
  Metric & Model & flchain & prostateSurvival & support \\
  \midrule
  $S_{\rm Cen-log-simple}$ & DeepHit $(\alpha=0)$ & $1.5059 \pm 0.0513$ & $1.3609 \pm 0.0301$ & $1.8296 \pm 0.0446$ \\
             & DeepHit $(\alpha=0.1)$    & $1.5200 \pm 0.0398$ & $1.3644 \pm 0.0293$ & $1.8481 \pm 0.0453$ \\
             & DeepHit $(\alpha=1)$  & $1.5858 \pm 0.0495$ & $1.3813 \pm 0.0318$ & $1.9996 \pm 0.0525$ \\
             & DeepHit $(\alpha=10)$ & $2.0313 \pm 0.1648$ & $1.5688 \pm 0.0823$ & $2.3657 \pm 0.0441$ \\
  \midrule
  D-calibration & DeepHit $(\alpha=0)$  & $0.0003 \pm 0.0001$ & $0.0001 \pm 0.0000$ & $0.0062 \pm 0.0012$ \\
             & DeepHit $(\alpha=0.1)$    & $0.0005 \pm 0.0002$ & $0.0001 \pm 0.0000$ & $0.0056 \pm 0.0009$ \\ 
             & DeepHit $(\alpha=1)$  & $0.0008 \pm 0.0003$ & $0.0003 \pm 0.0001$ & $0.0062 \pm 0.0010$ \\
             & DeepHit $(\alpha=10)$ & $0.0138 \pm 0.0046$ & $0.0064 \pm 0.0035$ & $0.0179 \pm 0.0053$ \\
  \midrule
  KM-calibration & DeepHit $(\alpha=0)$ & $0.0213 \pm 0.0049$ & $0.0343 \pm 0.0102$ & $0.0288 \pm 0.0127$ \\
             & DeepHit $(\alpha=0.1)$    & $0.0264 \pm 0.0071$ & $0.0418 \pm 0.0139$ & $0.0249 \pm 0.0067$ \\
             & DeepHit $(\alpha=1)$  & $0.0362 \pm 0.0084$ & $0.0599 \pm 0.0341$ & $0.0545 \pm 0.0110$ \\
             & DeepHit $(\alpha=10)$ & $0.2077 \pm 0.0543$ & $0.4937 \pm 0.1772$ & $0.4273 \pm 0.1188$ \\
  \end{tabular}
\end{table*}

Figure~\ref{fig:comparison} shows that the prediction performances of the four extended scoring rules (Cen-log, Cen-Brier, Cen-RPS, and Portnoy) were not similar, even though we prove that these four scoring rules are conditionally proper and the outputs are expected to be similar if the parameters $\hat{w}$ and $\{ \hat{w}_{i} \}_{i=0}^{B-1}$ are correct.  The scoring rules Cen-log and Cen-Brier outperformed the scoring rules Cen-RPS and Portnoy in discrimination performance $S_{\rm Cen-log-simple}$.   These results indicate that the accuracy of the estimated parameters $\hat{w}$ and $\{ \hat{w}_{i} \}_{i=0}^{B-1}$ by the IR algorithm are important when we use these scoring rules in practice.  The major difference between these scoring rules are that, whereas the set of parameters $\{ w_{i} \}_{i=0}^{B-1}$ in Cen-log and Cen-Brier usually satisfies $w_{i} \approx 0$ or $w_{i}=1$, the set of parameters $\{ w_{i} \}_{i=0}^{B}$ in Cen-RPS can take an arbitrary value $0 \leq w_{i} \leq 1$.  The parameter $w$ in Portnoy can also take an arbitrary value $0 \leq w \leq 1$.  Therefore, Cen-log and Cen-Brier seem less sensitive to the accuracy of the parameters than Cen-RPS and Portnoy.  This figure also shows that the other four scoring rules (DeepHit, DRSA, S-CRPS, and IPCW BS($t$)) performed worse than Cen-log and Cen-Brier.  Note that IPCW BS($t$) model is similar to the IR algorithm in that both of the algorithms are used to estimate unknown parameters, but the loss function of IPCW BS($t$) model is not proved to be proper in terms of the theory of scoring rules.  With respect to the calibration metrics, Cen-log and Cen-Brier showed comparable performance with the Kaplan-Meier estimator.  However, the other scoring rules showed worse calibration performances for at least one of D-calibration and KM-calibration.

Regarding the parameter $\alpha$ of DeepHit~\cite{LZYS18}, we conducted additional experiments by changing this parameter.  The loss function of DeepHit consists of two terms.  The first term is equal to the extension of the logarithmic score $S_{\rm Cen-log-simple}$, and the second term is used to improve a ranking metric (i.e., a variant of C-index).  The parameter $\alpha$ is used to control the balance between these two terms, and the weight for the second term is increased by using a large $\alpha$.  Note that the scoring rule $S_{\rm Cen-log-simple}$ is equivalent to DeepHit with $\alpha=0$.  Table~\ref{table:deephit_alpha} shows the results for $\alpha \in \{ 0, 0.1, 1, 10 \}$.  This table shows that the prediction performances of DeepHit became worse as $\alpha$ increases.  This means that we should set $\alpha=0$ when we use DeepHit.

\section{Conclusion}

We discussed extensions of four scoring rules for survival analysis, and we proved that these extensions are proper if the parameter $w$ or $\{ w_{i} \}_{i=0}^{B-1}$ is correct.  These proofs reduce the problem of estimating $\hat{F}$ to the problem of estimating the parameter $w$ or $\{ w_{i} \}_{i=0}^{B-1}$ in proper scoring rules.  We also demonstrated that the models with $S_{\rm Cen-log}$ and $S_{\rm Cen-Brier}$ as loss functions performed the best in our experiments.  These results indicate that it is better to use a proper scoring rule that has low sensitivity on the parameter.  In addition, we clarified the hidden assumption in the proof of the properness for $S_{\rm Cen-cont-log}$~\cite{RHSS22}.  This suggests us to use a sufficiently large $B$ when we use it, and we demonstrated that such $B$ can be found by comparing the prediction performances of $S_{\rm Cen-log-simple}$ and $S_{\rm Cen-log}$ with various $B$.


\bibliography{ScoringRule}

\begin{thebibliography}{42}
\providecommand{\natexlab}[1]{#1}
\providecommand{\url}[1]{\texttt{#1}}
\expandafter\ifx\csname urlstyle\endcsname\relax
  \providecommand{\doi}[1]{doi: #1}\else
  \providecommand{\doi}{doi: \begingroup \urlstyle{rm}\Url}\fi

\bibitem[Antolini et~al.(2005)Antolini, Boracchi, and Biganzoli]{ABB05}
Antolini, L., Boracchi, P., and Biganzoli, E.
\newblock A time-dependent discrimination index for survival data.
\newblock \emph{Statistics in Medicine}, 24\penalty0 (24):\penalty0 3927--3944,
  2005.

\bibitem[Avati et~al.(2019)Avati, Duan, Zhou, Jung, Shah, and Ng]{ADZJSN19}
Avati, A., Duan, T., Zhou, S., Jung, K., Shah, N.~H., and Ng, A.~Y.
\newblock Countdown regression: Sharp and calibrated survival predictions.
\newblock In \emph{Proceedings of UAI 2019}, pp.\  145--155, 2019.

\bibitem[Benedetti(2010)]{Ben10}
Benedetti, R.
\newblock Scoring rules for forecast verification.
\newblock \emph{American Meteorological Society}, 138\penalty0 (1):\penalty0
  203--211, 2010.

\bibitem[Bengs et~al.(2022)Bengs, H{\"u}llermeier, and Waegeman]{BHW22}
Bengs, V., H{\"u}llermeier, E., and Waegeman, W.
\newblock Pitfalls of epistemic uncertainty quantification through loss
  minimisation.
\newblock In \emph{Proceedings of NeurIPS 2022}, 2022.

\bibitem[Blanche et~al.(2018)Blanche, Kattan, and Gerds]{BKG18}
Blanche, P., Kattan, M.~W., and Gerds, T.~A.
\newblock The c-index is not proper for the evaluation of $t$-year predicted
  risks.
\newblock \emph{Biostatistics}, 20\penalty0 (2):\penalty0 347--357, 2018.

\bibitem[Brier(1950)]{Brier50}
Brier, G.~W.
\newblock Verification of forecasts expressed in terms of probability.
\newblock \emph{Monthly Weather Review}, 78\penalty0 (1):\penalty0 1--3, 1950.

\bibitem[Chapfuwa et~al.(2020)Chapfuwa, Tao, Li, Khan, Chandross, Pencina,
  Carin, and Henao]{CTL20}
Chapfuwa, P., Tao, C., Li, C., Khan, I., Chandross, K.~J., Pencina, M.~J.,
  Carin, L., and Henao, R.
\newblock Calibration and uncertainty in neural time-to-event modeling.
\newblock \emph{IEEE Transactions on Neural Networks and Learning Systems},
  34\penalty0 (4):\penalty0 1666--1680, 2020.

\bibitem[Cox(1972)]{Cox72}
Cox, D.~R.
\newblock Regression models and life-tables.
\newblock \emph{Journal of the Royal Statistical Society, Series B},
  34\penalty0 (2):\penalty0 187--220, 1972.

\bibitem[Dirick et~al.(2017)Dirick, Claeskens, and Baesens]{DCB17}
Dirick, L., Claeskens, G., and Baesens, B.
\newblock Time to default in credit scoring using survival analysis: a
  benchmark study.
\newblock \emph{Journal of the Operational Research Society}, 68\penalty0
  (6):\penalty0 652--665, 2017.

\bibitem[Dispenzieri et~al.(2012)Dispenzieri, Katzmann, Kyle, Larson, Therneau,
  Colby, Clark, Mead, Kumar, III, and Rajkumar]{DKKL12}
Dispenzieri, A., Katzmann, J.~A., Kyle, R.~A., Larson, D.~R., Therneau, T.~M.,
  Colby, C.~L., Clark, R.~J., Mead, G.~P., Kumar, S., III, L. J.~M., and
  Rajkumar, S.~V.
\newblock Use of nonclonal serum immunoglobulin free light chains to predict
  overall survival in the general population.
\newblock \emph{Mayo Clinic Proceedings}, 87\penalty0 (6):\penalty0 517--523,
  2012.

\bibitem[Gneiting \& Raftery(2007)Gneiting and Raftery]{GR07}
Gneiting, T. and Raftery, A.~E.
\newblock Strictly proper scoring rules, prediction, and estimation.
\newblock \emph{Journal of the American Statistical Association}, 102\penalty0
  (477):\penalty0 359--378, 2007.

\bibitem[Goldstein et~al.(2020)Goldstein, Han, Puli, Perotte, and
  Ranganath]{GHPPR20}
Goldstein, M., Han, X., Puli, A.~M., Perotte, A., and Ranganath, R.
\newblock {X-CAL}: Explicit calibration for survival analysis.
\newblock In \emph{Proceedings of NeurIPS 2020}, pp.\  18296--18307, 2020.

\bibitem[Good(1952)]{Good52}
Good, I.~J.
\newblock Rational decisions.
\newblock \emph{Journal of the Royal Statistical Society. Series B
  (Methodological)}, 14\penalty0 (1):\penalty0 107--114, 1952.

\bibitem[Graf et~al.(1999)Graf, Schmoor, Sauerbrei, and Schumacher]{GSSS99}
Graf, E., Schmoor, C., Sauerbrei, W., and Schumacher, M.
\newblock Assessment and comparison of prognostic classification schemes for
  survival data.
\newblock \emph{Statistics in Medicine}, 18\penalty0 (17--18):\penalty0
  2529--2545, 1999.

\bibitem[Haider et~al.(2020)Haider, Hoehn, Davis, and Greiner]{HHDG20}
Haider, H., Hoehn, B., Davis, S., and Greiner, R.
\newblock Effective ways to build and evaluate individual survival
  distributions.
\newblock \emph{Journal of Machine Learning Research}, 21\penalty0
  (85):\penalty0 1--63, 2020.

\bibitem[Han et~al.(2021)Han, Goldstein, Puli, Wies, Perotte, and
  Ranganath]{HGPWPR21}
Han, X., Goldstein, M., Puli, A., Wies, T., Perotte, A., and Ranganath, R.
\newblock Inverse-weighted survival games.
\newblock In \emph{Proceedings of NeurIPS 2021}, pp.\  2160--2172, 2021.

\bibitem[Harrell et~al.(1982)Harrell, Califf, Pryor, Lee, and Rosati]{HCPLR82}
Harrell, F.~E., Califf, R.~M., Pryor, D.~B., Lee, K.~L., and Rosati, R.~A.
\newblock Evaluating the yield of medical tests.
\newblock \emph{Journal of the American Medical Association}, 247\penalty0
  (18):\penalty0 2543--2546, 1982.

\bibitem[Kamran \& Wiens(2021)Kamran and Wiens]{KW21}
Kamran, F. and Wiens, J.
\newblock Estimating calibrated individualized survival curves with deep
  learning.
\newblock In \emph{Proceedings of AAAI 2021}, pp.\  240--248, 2021.

\bibitem[Kaplan \& Meier(1958)Kaplan and Meier]{KM58}
Kaplan, E.~L. and Meier, P.
\newblock Nonparametric estimation from incomplete observations.
\newblock \emph{Journal of the American Statistical Association}, 53\penalty0
  (282):\penalty0 457--481, 1958.

\bibitem[Kingma \& Ba(2015)Kingma and Ba]{KB15}
Kingma, D.~P. and Ba, J.
\newblock Adam: A method for stochastic optimization.
\newblock In \emph{Proceedings of ICLR 2015}, 2015.

\bibitem[Knaus et~al.(1995)Knaus, Harrell, Lynn, Goldman, Phillips, Connors,
  Dawson, Fulkerson, Califf, Desbiens, Layde, Oye, Bellamy, Hakim, and
  Wagner]{KHLG95}
Knaus, W.~A., Harrell, Jr., F.~E., Lynn, J., Goldman, L., Phillips, R.~S.,
  Connors, Jr., A.~F., Dawson, N.~V., Fulkerson, Jr., W.~J., Califf, R.~M.,
  Desbiens, N., Layde, P., Oye, R.~K., Bellamy, P.~E., Hakim, R.~B., and
  Wagner, D.~P.
\newblock The {SUPPORT} prognostic model. {O}bjective estimates of survival for
  seriously ill hospitalized adults. {S}tudy to understand prognoses and
  preferences for outcomes and risks of treatments.
\newblock \emph{Annals of Internal Medicine}, 122\penalty0 (3):\penalty0
  191--203, 1995.

\bibitem[Koenker \& Bassett(1978)Koenker and Bassett]{KB78}
Koenker, R. and Bassett, Jr., B.
\newblock Regression quantiles.
\newblock \emph{Econometrica}, 46\penalty0 (1):\penalty0 33--50, 1978.

\bibitem[Koenker \& Hallock(2001)Koenker and Hallock]{KH01}
Koenker, R. and Hallock, K.~F.
\newblock Quantile regression.
\newblock \emph{Journal of economic perspectives}, 15\penalty0 (4):\penalty0
  143--156, 2001.

\bibitem[Kvamme et~al.(2019)Kvamme, Borgan, and Scheel]{KBS19}
Kvamme, H., Borgan, O., and Scheel, I.
\newblock Time-to-event prediction with neural networks and {C}ox regression.
\newblock \emph{Journal of Machine Learning Research}, 20\penalty0
  (129):\penalty0 1--30, 2019.

\bibitem[Lee et~al.(2018)Lee, Zame, Yoon, and van~der Schaar]{LZYS18}
Lee, C., Zame, W.~R., Yoon, J., and van~der Schaar, M.
\newblock Deep{H}it: A deep learning approach to survival analysis with
  competing risks.
\newblock In \emph{Proceedings of AAAI-18}, pp.\  2314--2321, 2018.

\bibitem[Lu-Yao et~al.(2009)Lu-Yao, Albertsen, Moore, Shih, Lin, DiPaola,
  Barry, Zietman, O'Leary, Walker-Corkery, and Yao]{LAMS09}
Lu-Yao, G.~L., Albertsen, P.~C., Moore, D.~F., Shih, W., Lin, Y., DiPaola,
  R.~S., Barry, M.~J., Zietman, A., O'Leary, M., Walker-Corkery, E., and Yao,
  S.-L.
\newblock Outcomes of localized prostate cancer following conservative
  management.
\newblock \emph{Journal of the American Medical Association}, 302\penalty0
  (11):\penalty0 1202--1209, 2009.

\bibitem[Mura et~al.(2008)Mura, Galavotti, Hykel, and de~Finetti]{mura2008}
Mura, A., Galavotti, M., Hykel, H., and de~Finetti, B.
\newblock \emph{Philosophical Lectures on Probability: collected, edited, and
  annotated by Alberto Mura}.
\newblock Synthese Library. Springer Netherlands, 2008.

\bibitem[Neocleous et~al.(2006)Neocleous, Branden, and Portnoy]{NBP06}
Neocleous, T., Branden, K.~V., and Portnoy, S.
\newblock Correction to censored regression quantiles by {S}. {P}ortnoy, 98
  (2003), 1001--1012.
\newblock \emph{Journal of the American Statistical Association}, 101\penalty0
  (474):\penalty0 860--861, 2006.

\bibitem[Parmigiani \& Inoue(2009)Parmigiani and Inoue]{parmigiani2009}
Parmigiani, G. and Inoue, L.
\newblock \emph{Decision Theory: Principles and Approaches}.
\newblock Wiley Series in Probability and Statistics. Wiley, 2009.

\bibitem[Pearce et~al.(2022)Pearce, Jeong, Jia, and Zhu]{PJJZ22}
Pearce, T., Jeong, J.-H., Jia, Y., and Zhu, J.
\newblock Censored quantile regression neural networks.
\newblock In \emph{Proceedings of NeurIPS 2022}, 2022.

\bibitem[Peng(2021)]{Peng21}
Peng, L.
\newblock Quantile regression for survival data.
\newblock \emph{Annual Review of Statistics and Its Application}, 8:\penalty0
  413--437, 2021.

\bibitem[Portnoy(2003)]{Por03}
Portnoy, S.
\newblock Censored regression quantiles.
\newblock \emph{Journal of the American Statistical Association}, 98\penalty0
  (464):\penalty0 1001--1012, 2003.

\bibitem[{R Core Team}(2016)]{R}
{R Core Team}.
\newblock \emph{R: A Language and Environment for Statistical Computing}.
\newblock R Foundation for Statistical Computing, Vienna, Austria, 2016.
\newblock URL \url{https://www.R-project.org/}.

\bibitem[Ren et~al.(2019)Ren, Qin, Zheng, Yang, Zhang, Qiu, and Yu]{RQZYZQY19}
Ren, K., Qin, J., Zheng, L., Yang, Z., Zhang, W., Qiu, L., and Yu, Y.
\newblock Deep recurrent survival analysis.
\newblock In \emph{Proceedings of AAAI-19}, pp.\  4798--4805, 2019.

\bibitem[Rindt et~al.(2022)Rindt, Hu, Steinsaltz, and Sejdinovic]{RHSS22}
Rindt, D., Hu, R., Steinsaltz, D., and Sejdinovic, D.
\newblock Survival regression with proper scoring rules and monotonic neural
  networks.
\newblock In \emph{Proceedings of AISTATS 2022}, 2022.

\bibitem[Schlag et~al.(2015)Schlag, Tremewan, and van~der Weele]{STW15}
Schlag, K.~H., Tremewan, J., and van~der Weele, J.~J.
\newblock A penny for your thoughts: a survey of methods for eliciting beliefs.
\newblock \emph{Experimental Economics}, 18:\penalty0 457--490, 2015.

\bibitem[Sonabend et~al.(2022)Sonabend, Bender, and Vollmer]{SBV22}
Sonabend, R., Bender, A., and Vollmer, S.
\newblock Avoiding c-hacking when evaluating survival distribution predictions
  with discrimination measures.
\newblock \emph{Bioinformatics}, 38\penalty0 (17):\penalty0 4178--4184, 2022.

\bibitem[Tjandra et~al.(2021)Tjandra, He, and Wiens]{THW21}
Tjandra, D.~E., He, Y., and Wiens, J.
\newblock A hierarchical approach to multi-event survival analysis.
\newblock In \emph{Proceedings of AAAI 2021}, pp.\  591--599, 2021.

\bibitem[Uno et~al.(2011)Uno, Cai, Pencina, D'Agostino, and Wei]{UCPDW11}
Uno, H., Cai, T., Pencina, M.~J., D'Agostino, R.~B., and Wei, L.~J.
\newblock On the {C}-statistics for evaluating overall adequacy of risk
  prediction procedures with censored survival data.
\newblock \emph{Statistics in Medicine}, 30\penalty0 (10):\penalty0 1105--1117,
  2011.

\bibitem[Wang et~al.(2019)Wang, Li, and Reddy]{WLR19}
Wang, P., Li, Y., and Reddy, C.~K.
\newblock Machine learning for survival analysis: A survey.
\newblock \emph{ACM Computing Surveys}, 51\penalty0 (6):\penalty0 1--36, 2019.

\bibitem[Zheng et~al.(2019)Zheng, Yuan, and Wu]{ZYW19}
Zheng, P., Yuan, S., and Wu, X.
\newblock {SAFE}: A neural survival analysis model for fraud early detection.
\newblock In \emph{Proceedings of AAAI-19}, pp.\  1278--1285, 2019.

\bibitem[Zhong et~al.(2021)Zhong, Mueller, and Wang]{ZMW21}
Zhong, Q., Mueller, J., and Wang, J.-L.
\newblock Deep extended hazard models for survival analysis.
\newblock In \emph{Proceedings of NeurIPS 2021}, 2021.

\end{thebibliography}
\bibliographystyle{icml2023}

%%%%%%%%%%%%%%%%%%%%%%%%%%%%%%%%%%%%%%%%%%%%%%%%%%%%%%%%%%%%%%%%%%%%%%%%%%%%%%%
%%%%%%%%%%%%%%%%%%%%%%%%%%%%%%%%%%%%%%%%%%%%%%%%%%%%%%%%%%%%%%%%%%%%%%%%%%%%%%%
% APPENDIX
%%%%%%%%%%%%%%%%%%%%%%%%%%%%%%%%%%%%%%%%%%%%%%%%%%%%%%%%%%%%%%%%%%%%%%%%%%%%%%%
%%%%%%%%%%%%%%%%%%%%%%%%%%%%%%%%%%%%%%%%%%%%%%%%%%%%%%%%%%%%%%%%%%%%%%%%%%%%%%%
\newpage
\appendix
\onecolumn

\section{Proofs of Theorems}

We give proofs of the theorems, which are omitted from the main body of this paper.

\subsection{Portnoy's Estimator}\label{sec:proof_Portnoy}

We show a proof of Theorem~\ref{theorem:portnoy}.

\begin{proof}
We consider a fixed $c \sim C$, and we prove
\begin{equation}
\mathop{\mathbb{E}}_{t \sim T | C=c}[ S_{\rm Portnoy}(\hat{F}, (z, \delta); w, \tau) ] \geq \mathop{\mathbb{E}}_{t \sim T | C=c}[ S_{\rm Portnoy}(F, (z, \delta); w, \tau) ] \label{eq:proof_portnoy}
\end{equation}
for these four cases separately.
\begin{itemize}
\item{Case 1}: $c \leq \min \{ F^{-1}(\tau), \hat{F}^{-1}(\tau) \}$.
\item{Case 2}: $\max \{ F^{-1}(\tau), \hat{F}^{-1}(\tau) \} < c$.
\item{Case 3}: $F^{-1}(\tau) < c \leq \hat{F}^{-1}(\tau)$.
\item{Case 4}: $\hat{F}^{-1}(\tau) < c \leq F^{-1}(\tau)$.
\end{itemize}
Note that, if Inequality~(\ref{eq:proof_portnoy}) holds for any $c \sim C$, we can marginalize the inequality with respect to $C$, and we have
\[ \mathop{\mathbb{E}}_{t \sim T, c \sim C} [S_{\rm Portnoy}(\hat{F}, (z, \delta); w, \tau)] \geq \mathop{\mathbb{E}}_{t \sim T, c \sim C} [S_{\rm Portnoy}(F, (z,\delta); w, \tau)], \]
which means that $S_{\rm Portnoy}(\hat{F}, (z, \delta); w, \tau)$ is proper.  Therefore, we prove Inequality~(\ref{eq:proof_portnoy}) for the four cases.

\paragraph{Case 1.}

We prove the case for $c \leq \min \{ F^{-1}(\tau), \hat{F}^{-1}(\tau) \}$.  This means that $\tau_{c} \leq \tau$ and $w= (\tau - \tau_{c})/(1 - \tau_{c})$.  Hence, we have
\begin{eqnarray*}
S_{\rm Portnoy}(\hat{F}, (z, \delta); w, \tau)
& = &
\begin{cases}
\rho_{\tau}(\hat{F}^{-1}(\tau), t) & \mbox{if} \ t \leq c, \\
w \rho_{\tau}(\hat{F}^{-1}(\tau), c) + (1-w) \rho_{\tau}(\hat{F}^{-1}(\tau), z_{\infty}) & \mbox{if} \ t > c,
\end{cases} \\
& = &
\begin{cases}
(1 - \tau)(\hat{F}^{-1}(\tau) - t) & \mbox{if} \ t \leq c, \\
\frac{\tau - \tau_{c}}{1 - \tau_{c}}(1 - \tau)(\hat{F}^{-1}(\tau) - c) + \frac{1 - \tau}{1 - \tau_{c}} \tau(z_{\infty} - \hat{F}^{-1}(\tau)) & \mbox{if} \ t > c,
\end{cases} \\
& = &
\begin{cases}
(1 - \tau)(\hat{F}^{-1}(\tau) - t) & \mbox{if} \ t \leq c, \\
\frac{- \tau_{c}(1-\tau)}{1 - \tau_{c}} \hat{F}^{-1}(\tau) - \frac{(\tau-\tau_{c})(1 - \tau)}{1-\tau_{c}} c + \frac{(1 - \tau) \tau}{1 - \tau_{c}} z_{\infty} & \mbox{if} \ t > c.
\end{cases}
\end{eqnarray*}
By Assumption~\ref{assumption:independence}, we have $\mathrm{Pr}(t \leq c | C=c) = \mathrm{Pr}(t \leq c) = \tau_{c}$ and $\mathrm{Pr}(t > c | C=c) = 1 - \tau_{c}$.  Hence, we have
\begin{eqnarray*}
\mathop{\mathbb{E}}_{t \sim T | C=c}[ S_{\rm Portnoy}(\hat{F}, (z, \delta); w, \tau) ]
& = & {\rm Pr}(t \leq c | C=c) (1 - \tau) \hat{F}^{-1}(\tau) - (1-\tau) \mathop{\mathbb{E}}_{t \sim T | C=c, t \leq c}[t] \\
& & - {\rm Pr}(t > c | C=c) \frac{\tau_{c}(1-\tau)}{1 - \tau_{c}} \hat{F}^{-1}(\tau) - \frac{(\tau-\tau_{c})(1 - \tau)}{1-\tau_{c}} c + \frac{(1 - \tau) \tau}{1 - \tau_{c}} z_{\infty} \\
& = &  - (1-\tau) \mathop{\mathbb{E}}_{t \sim T | C=c, t \leq c}[t] - \frac{(\tau-\tau_{c})(1 - \tau)}{1-\tau_{c}} c + \frac{(1 - \tau) \tau}{1 - \tau_{c}} z_{\infty}.
\end{eqnarray*}
Since this value is the same for $S_{\rm Portnoy}(\hat{F}, (z, \delta); w, \tau)$ and $S_{\rm Portnoy}(F, (z, \delta); w, \tau)$, we have
\[ \mathop{\mathbb{E}}_{t \sim T | C=c}[ S_{\rm Portnoy}(\hat{F}, (z, \delta); w, \tau) ] = \mathop{\mathbb{E}}_{t \sim T | C=c}[ S_{\rm Portnoy}(F, (z, \delta); w, \tau) ]. \]

\paragraph{Case 2.}

We prove the case for $\max \{ F^{-1}(\tau), \hat{F}^{-1}(\tau) \} < c$.

If $F^{-1}(\tau) \leq \hat{F}^{-1}(\tau) < c$, then we have
\begin{eqnarray*}
S_{\rm Portnoy}(\hat{F}, (z, \delta); w, \tau)
& = &
\begin{cases}
\rho_{\tau}(\hat{F}^{-1}(\tau), t) & \mbox{if} \ t \leq c, \\
w \rho_{\tau}(\hat{F}^{-1}(\tau), c) + (1-w) \rho_{\tau}(\hat{F}^{-1}(\tau), z_{\infty}) & \mbox{if} \ t > c,
\end{cases} \\
& = &
\begin{cases}
(1 - \tau) (\hat{F}^{-1}(\tau) - t) & \mbox{if} \ t \leq \hat{F}^{-1}(\tau), \\
- \tau (\hat{F}^{-1}(\tau) - t) & \mbox{if} \ \hat{F}^{-1}(\tau) < t \leq c, \\
- w \tau (\hat{F}^{-1}(\tau) - c) - (1-w) \tau (\hat{F}^{-1}(\tau) - z_{\infty}) & \mbox{if} \ t > c,
\end{cases} \\
& \geq &
\begin{cases}
(1 - \tau) (\hat{F}^{-1}(\tau) - t) & \mbox{if} \ t \leq F^{-1}(\tau), \\
- \tau (\hat{F}^{-1}(\tau) - t) & \mbox{if} \ F^{-1}(\tau) < t \leq c, \\
- \tau \hat{F}^{-1}(\tau) + w \tau c + (1-w) z_{\infty} & \mbox{if} \ t > c,
\end{cases}
\end{eqnarray*}
where we used $-\tau(\hat{F}^{-1}(\tau) - t) \leq (1-\tau)(\hat{F}^{-1}(\tau) - t)$ when $F^{-1}(\tau) < t \leq \hat{F}^{-1}(\tau)$ for the inequality.

If $\hat{F}^{-1}(\tau) \leq F^{-1}(\tau) < c$, then we have
\begin{eqnarray*}
S_{\rm Portnoy}(\hat{F}, (z, \delta); w, \tau)
& = &
\begin{cases}
\rho_{\tau}(\hat{F}^{-1}(\tau), t) & \mbox{if} \ t \leq c, \\
w \rho_{\tau}(\hat{F}^{-1}(\tau), c) + (1-w) \rho_{\tau}(\hat{F}^{-1}(\tau), z_{\infty}) & \mbox{if} \ t > c,
\end{cases} \\
& = &
\begin{cases}
(1 - \tau) (\hat{F}^{-1}(\tau) - t) & \mbox{if} \ t \leq \hat{F}^{-1}(\tau), \\
- \tau (\hat{F}^{-1}(\tau) - t) & \mbox{if} \ \hat{F}^{-1}(\tau) < t \leq c, \\
- w \tau (\hat{F}^{-1}(\tau) - c) - (1-w) \tau (\hat{F}^{-1}(\tau) - z_{\infty}) & \mbox{if} \ t > c,
\end{cases} \\
& > &
\begin{cases}
(1 - \tau) (\hat{F}^{-1}(\tau) - t) & \mbox{if} \ t \leq F^{-1}(\tau), \\
- \tau (\hat{F}^{-1}(\tau) - t) & \mbox{if} \ F^{-1}(\tau) < t \leq c, \\
- \tau \hat{F}^{-1}(\tau) + w \tau c + (1-w) z_{\infty} & \mbox{if} \ t > c,
\end{cases}
\end{eqnarray*}
where we used $-\tau(\hat{F}^{-1}(\tau) - t) > (1-\tau)(\hat{F}^{-1}(\tau) - t)$ when $\hat{F}^{-1}(\tau) < t \leq F^{-1}(\tau)$ for the inequality.

By Assumption~\ref{assumption:independence}, we have $\mathrm{Pr}(t \leq F^{-1}(\tau) | C=c) = \mathrm{Pr}(t \leq F^{-1}(\tau)) = \tau$, $\mathrm{Pr}(F^{-1}(\tau) < t | C=c) = 1 - \tau$, and $\mathrm{Pr}(c < t | C=c) = 1 - \tau_{c}$.  Hence, we have
\begin{eqnarray*}
\mathop{\mathbb{E}}_{t \sim T | C=c}[ S_{\rm Portnoy}(\hat{F}, (z, \delta); w, \tau) ]
& \geq & {\rm Pr}(t \leq F^{-1}(\tau) | C=c) (1 - \tau) \hat{F}^{-1}(\tau) - (1-\tau) \mathop{\mathbb{E}}_{t \sim T | C=c, t \leq F^{-1}(\tau)}[t] \\
& & - {\rm Pr}(F^{-1}(\tau) < t | C=c) \tau \hat{F}^{-1}(\tau) \\
& & + {\rm Pr}(c < t | C=c) (w \tau c + (1-w) z_{\infty}) \\
& = & - (1-\tau) \mathop{\mathbb{E}}_{t \sim T | C=c, t \leq F^{-1}(\tau)}[t] + (1-\tau_{c})(w \tau c + (1-w) z_{\infty}).
\end{eqnarray*}
By using a similar argument, we have
\[
\mathop{\mathbb{E}}_{t \sim T | C=c}[ S_{\rm Portnoy}(F, (z, \delta); w, \tau) ]
= - (1-\tau) \mathop{\mathbb{E}}_{t \sim T | C=c, t \leq F^{-1}(\tau)}[t] + (1-\tau_{c})(w \tau c + (1-w) z_{\infty}).
\]
Note that this equation holds with equality.

Hence, we have
\[ \mathop{\mathbb{E}}_{t \sim T | C=c}[ S_{\rm Portnoy}(\hat{F}, (z, \delta); w, \tau) ] \geq \mathop{\mathbb{E}}_{t \sim T | C=c}[ S_{\rm Portnoy}(F, (z, \delta); w, \tau) ]. \]

\paragraph{Case 3.}

We prove the case for $F^{-1}(\tau) < c \leq \hat{F}^{-1}(\tau)$.

We have
\begin{eqnarray*}
S_{\rm Portnoy}(\hat{F}, (z, \delta); w, \tau)
& = &
\begin{cases}
\rho_{\tau}(\hat{F}^{-1}(\tau), t) & \mbox{if} \ t \leq c, \\
w \rho_{\tau}(\hat{F}^{-1}(\tau), c) + (1-w) \rho_{\tau}(\hat{F}^{-1}(\tau), z_{\infty}) & \mbox{if} \ t > c, \\
\end{cases} \\
& = &
\begin{cases}
(1 - \tau) (\hat{F}^{-1}(\tau) - t) & \mbox{if} \ t \leq c, \\
w (1 - \tau) (\hat{F}^{-1}(\tau) - c) - (1-w) \tau (\hat{F}^{-1}(\tau) - z_{\infty}) & \mbox{if} \  t > c, \\
\end{cases} \\
& \geq &
\begin{cases}
(1 - \tau) (\hat{F}^{-1}(\tau) - t) & \mbox{if} \ t \leq F^{-1}(\tau), \\
- \tau (\hat{F}^{-1}(\tau) - t) & \mbox{if} \ F^{-1}(\tau) < t \leq c, \\
- w \tau (\hat{F}^{-1}(\tau) - c) - (1-w) \tau (\hat{F}^{-1}(\tau) - z_{\infty}) & \mbox{if} \  t > c, \\
\end{cases} \\
\end{eqnarray*}
where we used $(1 - \tau) (\hat{F}^{-1}(\tau) - t) \geq - \tau (\hat{F}^{-1}(\tau) - t)$ when $F^{-1}(\tau) < t \leq c$ and $w (1 - \tau) (\hat{F}^{-1}(\tau) - c) \geq - w \tau (\hat{F}^{-1}(\tau) - c)$ when $t > c$.

By using a similar argument, we have
\begin{eqnarray*}
S_{\rm Portnoy}(F, (z, \delta); w, \tau)
& = &
\begin{cases}
\rho_{\tau}(F^{-1}(\tau), t) & \mbox{if} \ t \leq c, \\
w \rho_{\tau}(F^{-1}(\tau), c) + (1-w) \rho_{\tau}(F^{-1}(\tau), z_{\infty}) & \mbox{if} \ t > c, \\
\end{cases} \\
& = &
\begin{cases}
(1 - \tau) (F^{-1}(\tau) - t) & \mbox{if} \ t \leq F^{-1}(\tau), \\
- \tau (F^{-1}(\tau) - t) & \mbox{if} \ F^{-1}(\tau) < t \leq c, \\
- w \tau (F^{-1}(\tau) - c) - (1-w) \tau (F^{-1}(\tau) - z_{\infty}) & \mbox{if} \  t > c, \\
\end{cases}
\end{eqnarray*}
Note that this equation holds with equality.

Hence, we have
\[ \mathop{\mathbb{E}}_{t \sim T | C=c}[ S_{\rm Portnoy}(\hat{F}, (z, \delta); w, \tau) ] \geq \mathop{\mathbb{E}}_{t \sim T | C=c}[ S_{\rm Portnoy}(F, (z, \delta); w, \tau) ]. \]

\paragraph{Case 4.}

We prove the case for $\hat{F}^{-1}(\tau) < c \leq F^{-1}(\tau)$. 

Regarding $\hat{F}$, we have
\begin{eqnarray*}
S_{\rm Portnoy}(\hat{F}, (z, \delta); w, \tau)
& = &
\begin{cases}
\rho_{\tau}(\hat{F}^{-1}(\tau), t) & \mbox{if} \ t \leq c, \\
w \rho_{\tau}(\hat{F}^{-1}(\tau), c) + (1-w) \rho_{\tau}(\hat{F}^{-1}(\tau), z_{\infty}) & \mbox{if} \ t > c, \\
\end{cases} \\
& = &
\begin{cases}
(1 - \tau) (\hat{F}^{-1}(\tau) - t) & \mbox{if} \ t \leq \hat{F}^{-1}(\tau), \\
- \tau (\hat{F}^{-1}(\tau) - t) & \mbox{if} \ \hat{F}^{-1}(\tau) < t \leq c, \\
- w \tau (\hat{F}^{-1}(\tau) - c) - (1-w) \tau (\hat{F}^{-1}(\tau) - z_{\infty}) & \mbox{if} \  t > c, \\
\end{cases} \\
& > &
\begin{cases}
(1 - \tau) (\hat{F}^{-1}(\tau) - t) & \mbox{if} \ t \leq c, \\
- \tau \hat{F}^{-1}(\tau) + w \tau c + (1-w) \tau z_{\infty} & \mbox{if} \  t > c, \\
\end{cases}
\end{eqnarray*}
where we used $-\tau(\hat{F}^{-1}(\tau) - t) > (1-\tau)(\hat{F}^{-1}(\tau) - t)$ when $\hat{F}^{-1}(\tau) < t \leq c$ for the inequality. By Assumption~\ref{assumption:independence}, we have $\mathrm{Pr}(t \leq c | C=c) = \mathrm{Pr}(t \leq c) = \tau_{c}$ and $\mathrm{Pr}(t > c | C=c) = 1 - \tau_{c}$.  Hence, we have
\begin{eqnarray*}
\mathop{\mathbb{E}}_{t \sim T | C=c}[ S_{\rm Portnoy}(\hat{F}, (z, \delta); w, \tau) ]
& > & {\rm Pr}(t \leq c | C=c) (1 - \tau) \hat{F}^{-1}(\tau) - (1-\tau) \mathop{\mathbb{E}}_{t \sim T | C=c, t \leq c}[t] \\
& & + {\rm Pr}(t > c | C=c) (-\tau \hat{F}^{-1}(\tau) + w \tau c + (1-w) \tau z_{\infty}) \\
& > & \tau_{c} (1 - \tau) \hat{F}^{-1}(\tau) - (1-\tau) \mathop{\mathbb{E}}_{t \sim T | C=c, t \leq c}[t] \\
& & - (1-\tau_{c}) \tau \hat{F}^{-1}(\tau) + (1-\tau_{c})(w \tau c + (1-w) \tau z_{\infty}) \\
& > & (\tau_{c} - \tau) \hat{F}^{-1}(\tau) - (1-\tau) \mathop{\mathbb{E}}_{t \sim T | C=c, t \leq c}[t] + (1-\tau_{c})(w \tau c + (1-w) \tau z_{\infty}).
\end{eqnarray*}

Regarding $F$, since $w=(\tau - \tau_{c})/(1 - \tau_{c})$, we have
\begin{eqnarray*}
S_{\rm Portnoy}(F, (z, \delta); w, \tau)
& = &
\begin{cases}
\rho_{\tau}(F^{-1}(\tau), t) & \mbox{if} \ t \leq c, \\
w \rho_{\tau}(F^{-1}(\tau), c) + (1-w) \rho_{\tau}(F^{-1}(\tau), z_{\infty}) & \mbox{if} \ t > c, \\
\end{cases} \\
& = &
\begin{cases}
(1 - \tau) (F^{-1}(\tau) - t) & \mbox{if} \ t \leq c, \\
w (1 - \tau) (F^{-1}(\tau) - c) - (1-w) \tau (F^{-1}(\tau) - z_{\infty}) & \mbox{if} \  t > c, \\
\end{cases} \\
& = &
\begin{cases}
(1 - \tau) (F^{-1}(\tau) - t) & \mbox{if} \ t \leq c, \\
- \frac{\tau_{c}(1 - \tau)}{1 - \tau_{c}} F^{-1}(\tau) - w(1-\tau)c + (1-w) \tau z_{\infty} & \mbox{if} \ t > c, \\
\end{cases}
\end{eqnarray*}
By Assumption~\ref{assumption:independence}, we have $\mathrm{Pr}(t \leq c | C=c) = \mathrm{Pr}(t \leq c) = \tau_{c}$ and $\mathrm{Pr}(t > c | C=c) = 1 - \tau_{c}$.  Hence, we have
\begin{eqnarray*}
\mathop{\mathbb{E}}_{t \sim T | C=c}[ S_{\rm Portnoy}(\hat{F}, (z, \delta); w, \tau) ]
& = & {\rm Pr}(t \leq c | C=c) (1 - \tau) \hat{F}^{-1}(\tau) - (1-\tau) \mathop{\mathbb{E}}_{t \sim T | C=c, t \leq c}[t] \\
& & + {\rm Pr}(t > c | C=c) (-\frac{\tau_{c}(1 - \tau)}{1 - \tau_{c}} F^{-1}(\tau) - w(1-\tau)c + (1-w) \tau z_{\infty}) \\
& = & \tau_{c} (1 - \tau) \hat{F}^{-1}(\tau) - (1-\tau) \mathop{\mathbb{E}}_{t \sim T | C=c, t \leq c}[t] \\
& & - \tau_{c} (1-\tau) \hat{F}^{-1}(\tau) + (1-\tau_{c})(-w (1-\tau) c + (1-w) \tau z_{\infty}) \\
& = &  - (1-\tau) \mathop{\mathbb{E}}_{t \sim T | C=c, t \leq c}[t] + (1-\tau_{c})(-w (1-\tau) c + (1-w) \tau z_{\infty}).
\end{eqnarray*}

Therefore, since $\tau_{c} \leq \tau$ and $w=(\tau - \tau_{c})/(1 - \tau_{c})$, we have
\begin{eqnarray*}
\lefteqn{\mathop{\mathbb{E}}_{t \sim T | C=c}[ S_{\rm Portnoy}(\hat{F}, (z, \delta); w, \tau) ] - \mathop{\mathbb{E}}_{t \sim T | C=c}[ S_{\rm Portnoy}(F, (z, \delta); w, \tau) ]} \\
& > & ((\tau_{c} - \tau) \hat{F}^{-1}(\tau) + (1-\tau_{c}) w \tau c) + (1-\tau_{c})w (1-\tau) c \\
%& = & (\tau_{c} - \tau) \hat{F}^{-1}(\tau) + (1-\tau_{c}) w c \\
& = & (\tau_{c} - \tau) (\hat{F}^{-1}(\tau) - c) \\
& \geq & 0.
\end{eqnarray*}

\end{proof}

\subsection{Extension of Logarithmic Score}\label{sec:variant_logarithmic}

We show a proof of Theorem~\ref{theorem:logarithmic}.

\begin{proof}
We consider a fixed $c \sim C$, and let $t$ be a sample obtained from $T$.  Let $i$ be the index such that $\zeta_{i} \leq c < \zeta_{i+1}$.  Since Assumption~\ref{assumption:independence} holds, we have $\mathrm{Pr}(\zeta_{j} < t \leq \zeta_{j+1} | C=c) = \mathrm{Pr}(\zeta_{j} < t \leq \zeta_{j+1}) = F(\zeta_{j+1}) - F(\zeta_{j}) = f_{j}$ for any $j < i$, $\mathrm{Pr}(\zeta_{i} < t \leq c | C=c) = F(c) - F(\zeta_{i})$, and $\mathrm{Pr}(c < t| C=c) = \mathrm{Pr}(c < t) = 1 - F(c)$.  Hence, we have
\begin{eqnarray*}
\mathop{\mathbb{E}}_{t \sim T | C=c} [S_{\rm Cen-log}(\hat{F}, (z, \delta); \{ w_{k} \}_{k=0}^{B-1}, \{ \zeta_{k} \}_{k=0}^{B})]
& = & - \sum_{j < i} \mathrm{Pr}(\zeta_{j} < t \leq \zeta_{j+1} | C=c) \log \hat{f}_{j} \\
& & - \mathrm{Pr}(\zeta_{i} < t \leq c | C=c) \log \hat{f}_{i} \\
& & - \mathrm{Pr}(c < t | C=c) \left( w_{i} \log \hat{f}_{i} + (1-w_{i}) \log(1 - \hat{F}(\zeta_{i+1})) \right) \\
& = & - \sum_{j < i} f_{j} \log \hat{f}_{j} \\
& & - (F(c) - F(\zeta_{i})) \log \hat{f}_{i} \\
& & - (1 - F(c)) \left( w_{i} \log \hat{f}_{i} + (1-w_{i}) \log(1 - \hat{F}(\zeta_{i+1})) \right) \\
& = & - \sum_{j \leq i} f_{j} \log \hat{f}_{j} - (1 - F(\zeta_{i+1})) \log(1 - \hat{F}(\zeta_{i+1})),
\end{eqnarray*}
where we used $w_{i} = (F(\zeta_{i+1}) - F(c))/(1 - F(c))$ for the last equality.

Hence, we have
\begin{eqnarray}
\lefteqn{\mathop{\mathbb{E}}_{t \sim T | C=c} [S_{\rm Cen-log}(\hat{F}, (z, \delta); \{ w_{k} \}_{k=0}^{B-1}, \{ \zeta_{k} \}_{k=0}^{B})] - \mathop{\mathbb{E}}_{t \sim T | C=c} [S_{\rm Cen-log}(F, (z, \delta); \{ w_{k} \}_{k=0}^{B-1}, \{ \zeta_{k} \}_{k=0}^{B})]} \nonumber \\
& = & - \sum_{j \leq i} f_{j} (\log \hat{f}_{j} - \log f_{j}) - (1 - F(\zeta_{i+1})) (\log(1 - \hat{F}(\zeta_{i+1}))- \log(1 - F(\zeta_{i+1}))) \nonumber \\
& \geq & 0,  \label{ineq:logarithmic}
\end{eqnarray}
where we used the fact that the Kullback-Leibler divergence between two probability distributions is non-negative for the inequality.  This means that the inequality 
\[ - \sum_{k} p_{k} (\log \hat{p}_{k} - \log p_{k}) \geq 0 \]
holds for any two probability distributions $p_{k}$ and $\hat{p}_{k}$ and the equality holds only if $p_{k}=\hat{p}_{k}$ for all $k$.  Here, we use an $(i+2)$-dimensional vector $\bm{p}=(p_{0}, p_{1}, \ldots, p_{i+1})$; we set $p_{k} = f_{k}$ for all $k \leq i$ and we set $p_{i+1} = 1 - F(\zeta_{i+1})$.  Note that the vectors $\bm{p}$ and $\hat{\bm{p}}$ constructed in this way are a probability distribution (i.e., $\sum_{k} p_{k} = 1$).

Since Inequality~(\ref{ineq:logarithmic}) holds for any $c \sim C$, we marginalize the inequality with respect to $C$, and we have
\[ \mathop{\mathbb{E}}_{t \sim T, c \sim C} [S_{\rm Cen-log}(\hat{F}, (z, \delta); \{ w_{i} \}_{i=0}^{B-1}, \{ \zeta_{i} \}_{i=0}^{B})] \geq \mathop{\mathbb{E}}_{t \sim T, c \sim C} [S_{\rm Cen-log}(F, (z,\delta); \{ w_{i} \}_{i=0}^{B-1}, \{ \zeta_{i} \}_{i=0}^{B})], \]
which means that $S_{\rm Cen-log}(\hat{F}, (z, \delta); \{ w_{i} \}_{i=0}^{B-1}, \{ \zeta_{i} \}_{i=0}^{B})$ is proper.
\end{proof}

\subsection{Extension of Brier Score}\label{sec:variant_Brier}

We show a proof of Theorem~\ref{theorem:Brier}.

\begin{proof}
We consider a fixed $c \sim C$, and let $t$ be a sample obtained from $T$.  Let $j$ be the index such that $\zeta_{j} < c \leq \zeta_{j+1}$.  Since Assumption~\ref{assumption:independence} holds, we have $\mathrm{Pr}(\zeta_{i} < t \leq \zeta_{i+1} | C=c) = \mathrm{Pr}(\zeta_{i} < t \leq \zeta_{i+1}) = F(\zeta_{i+1}) - F(\zeta_{i}) = f_{i}$ for any $i < j$, $\mathrm{Pr}(\zeta_{j} < t \leq c | C=c) = F(c) - F(\zeta_{j})$, and $\mathrm{Pr}(c < t| C=c) = \mathrm{Pr}(c < t) = 1 - F(c)$.  Hence, we have
\begin{eqnarray*}
\lefteqn{\mathop{\mathbb{E}}_{t \sim T | C=c} [S_{\rm Cen-Brier}(\hat{F}, (z, \delta); \{ w_{k} \}_{k=0}^{B-1}, \{ \zeta_{k} \}_{k=0}^{B})]} \\
& = & \sum_{i < j} \mathrm{Pr}(\zeta_{i} < t \leq \zeta_{i+1} | C=c) \left( (1- \hat{f}_{i})^2 + \sum_{k \neq i} \hat{f}_{k}^2 \right) \\
& & + \mathrm{Pr}(\zeta_{j} < t \leq c | C=c)  \left( (1- \hat{f}_{j})^2 + \sum_{k \neq j} \hat{f}_{k}^2 \right) \\
& & + \mathrm{Pr}(c < t | C=c) \left( w_{j} (1- \hat{f}_{j})^2 + (1-w_{j}) \hat{f}_{j}^2 + \sum_{i < j} \hat{f}_{i}^2 + \sum_{i > j} \left( w_{i} (1-\hat{f}_{i})^2 + (1-w_{i}) \hat{f}_{i}^{2} \right) \right) \\
& = & \sum_{i < j} f_{i} \left( (1- \hat{f}_{i})^2 + \sum_{k \neq i} \hat{f}_{k}^2 \right) + (F(c) - F(\zeta_{j}))  \left( (1- \hat{f}_{j})^2 + \sum_{k \neq j} \hat{f}_{k}^2 \right) \\
& & + (1 - F(c)) \left( w_{j} (1- \hat{f}_{j})^2 + (1-w_{j}) \hat{f}_{j}^2 + \sum_{i < j} \hat{f}_{i}^2 + \sum_{i > j} \left( w_{i} (1-\hat{f}_{i})^2 + (1-w_{i}) \hat{f}_{i}^{2} \right) \right) \\
& = & \sum_{i} \left( f_{i} (1- \hat{f}_{i})^2 + (1-f_{i}) \hat{f}_{i}^2 \right) \\
& = & \sum_{i} (\hat{f}_{i}^2 - 2 f_{i} \hat{f}_{i} + f_{i}),
\end{eqnarray*}
where we used
\[
w_{i} =
\begin{cases}
0 & \mbox{if} \ \delta=1 \ \mbox{and} \ \zeta_{i+1} < z=t, \\
1 & \mbox{if} \ \delta=1 \ \mbox{and} \ \zeta_{i} < z = t \leq \zeta_{i+1}, \\
0 & \mbox{if} \ z \leq \zeta_{i} \\
\end{cases}
\]
for the first equality and
\[
w_{i} =
\begin{cases}
 (F(\zeta_{i+1}) - F(c))/(1 - F(c)) & \mbox{if} \ \delta=0 \ \mbox{and} \ i=j, \\
 f_{i}/(1-F(c)) & \mbox{if} \ \delta=0 \ \mbox{and} \ i > j
\end{cases}
\]
for the last equality.

Hence we have
\begin{eqnarray}
\lefteqn{\mathop{\mathbb{E}}_{t \sim T | C=c} [S_{\rm Cen-Brier}(\hat{F}, (z, \delta); \{ w_{i} \}_{i=0}^{B-1}, \{ \zeta_{i} \}_{i=0}^{B})] - \mathop{\mathbb{E}}_{t \sim T | C=c} [S_{\rm Cen-Brier}(F, (z, \delta); \{ w_{i} \}_{i=0}^{B-1}, \{ \zeta_{i} \}_{i=0}^{B})]} \nonumber \\
& = & \sum_{i} (\hat{f}_{i}^2 - f_{i}^2 - 2 f_{i} (\hat{f}_{i} - f_{i})) \qquad \qquad \qquad \qquad \qquad \qquad \qquad \qquad \qquad \qquad \qquad \qquad \qquad \qquad \nonumber \\
& = & \sum_{i} (\hat{f}_{i} - f_{i})^2 \nonumber \\
& \geq & 0.  \label{ineq:Brier}
\end{eqnarray}
Note that the equality holds only if $\hat{f}_{i} = f_{i}$ holds for all $i$.

Since Inequality~(\ref{ineq:Brier}) holds for any $c \sim C$, we have
\[
\mathop{\mathbb{E}}_{t \sim T, c \sim C} [S_{\rm Cen-Brier}(\hat{F}, (z, \delta); \{ w_{i} \}_{i=0}^{B-1}, \{ \zeta_{i} \}_{i=0}^{B})] \geq \mathop{\mathbb{E}}_{t \sim T, c \sim C} [S_{\rm Cen-Brier}(F, (z,\delta); \{ w_{i} \}_{i=0}^{B-1}, \{ \zeta_{i} \}_{i=0}^{B})],
\]
which means that $S_{\rm Cen-Brier}(\hat{F}, (z, \delta); \{ w_{i} \}_{i=0}^{B-1}, \{ \zeta_{i} \}_{i=0}^{B})$ is proper.
\end{proof}

\section{Additional Experiments}\label{sec:experiments_appendix}

We investigated the differences of the prediction performances between $S_{\rm Cen-log}$ (defined in Eq.~(\ref{eq:logarithmic_censored})) and $S_{\rm Cen-log-simple}$ (defined in Eq.~(\ref{eq:logarithmic_simple})) by using $S_{\rm Cen-log-simple}$, D-calibration, and KM-calibration as metrics to determine the parameter $B$. Tables~\ref{table:B8}--~\ref{table:B32} show the results for $B=8, 16, 32$, respectively, where each number shows the mean and variance of the values obtained by five random runs and the bold numbers were used to emphasize the difference between two scoring rules. These results show that the prediction performances of these two scoring rules were similar for the prostateSurvival and support datasets even for $B=8$.  However, they showed different prediction performances for the flchain dataset for $B=8$ and $B=16$, but the performance difference was negligible for $B=32$.  Therefore, we used $B=32$ in the other experiments in this paper.

\begin{table}[t]
  \caption{Comparison between two extensions of logarithmic score for $B=8$.}
  \label{table:B8}
  \centering
  \begin{tabular}{ll|lll}
  Metric & Loss Function & flchain & prostateSurvival & support \\
  \midrule
  $S_{\rm Cen-log-simple}$ & $S_{\rm Cen-log}$ & $6.4618 \pm 0.1204$ & $1.3460 \pm 0.0476$ & $1.5422 \pm 0.0704$ \\
                           &$S_{\rm Cen-log-simple}$ & $6.4176 \pm 0.1266$ & $1.3447 \pm 0.0451$ & $1.5368 \pm 0.0701$ \\
  \midrule
  D-calibration & $S_{\rm Cen-log}$ & $\bm{0.0045} \pm 0.0004$ & $0.0002 \pm 0.0000$ & $0.0370 \pm 0.0032$ \\
           & $S_{\rm Cen-log-simple}$ & $\bm{0.0127} \pm 0.0013$ & $0.0002 \pm 0.0001$ & $0.0349 \pm 0.0024$ \\
  \midrule
  KM-calibration & $S_{\rm Cen-log}$ & $\bm{0.0048} \pm 0.0026$ & $0.0048 \pm 0.0028$ & $0.0057 \pm 0.0027$ \\
             & $S_{\rm Cen-log-simple}$ & $\bm{0.0614} \pm 0.0081$ & $0.0083 \pm 0.0024$ & $0.0061 \pm 0.0033$ \\
   \\
  \end{tabular}
  \caption{Comparison between two extensions of logarithmic score for $B=16$.}
  \label{table:B16}
  \centering
  \begin{tabular}{ll|lll}
  Metric & Loss Function & flchain & prostateSurvival & support \\
  \midrule
  $S_{\rm Cen-log-simple}$ & $S_{\rm Cen-log}$ & $3.6774 \pm 0.0386$ & $1.2880 \pm 0.0247$ & $1.6017 \pm 0.0733$ \\
                          & $S_{\rm Cen-log-simple}$ & $3.6676 \pm 0.0424$ & $1.3447 \pm 0.0451$ & $1.6008 \pm 0.0731$ \\
  \midrule
  D-calibration & $S_{\rm Cen-log}$ & $\bm{0.0005} \pm 0.0002$ & $0.0001 \pm 0.0000$ & $0.0147 \pm 0.0020$ \\
          & $S_{\rm Cen-log-simple}$ & $\bm{0.0013} \pm 0.0004$ & $0.0002 \pm 0.0000$ & $0.0143 \pm 0.0021$ \\
  \midrule
  KM-calibration & $S_{\rm Cen-log}$ & $0.0117 \pm 0.0046$ & $0.0142 \pm 0.0036$ & $0.0149 \pm 0.0080$ \\
             & $S_{\rm Cen-log-simple}$ & $0.0162 \pm 0.0049$ & $0.0158 \pm 0.0063$ & $0.0158 \pm 0.0100$ \\
  \\
  \end{tabular}
  \caption{Comparison between two extensions of logarithmic score for $B=32$.}
  \label{table:B32}
  \centering
  \begin{tabular}{ll|lll}
  Metric & Loss Function & flchain & prostateSurvival & support \\
  \midrule
  $S_{\rm Cen-log-simple}$ & $S_{\rm Cen-log}$ & $1.5054 \pm 0.0508$ & $1.3608 \pm 0.0295$ & $1.8307 \pm 0.0452$ \\
                           & $S_{\rm Cen-log-simple}$ & $1.5059 \pm 0.0513$ & $1.3609 \pm 0.0301$ & $1.8296 \pm 0.0446$ \\
  \midrule
  D-calibration & $S_{\rm Cen-log}$ & $0.0003 \pm 0.0001$ & $0.0001 \pm 0.0000$ & $0.0063 \pm 0.0009$ \\
          & $S_{\rm Cen-log-simple}$ & $0.0003 \pm 0.0001$ & $0.0001 \pm 0.0000$ & $0.0062 \pm 0.0012$ \\
  \midrule
  KM-calibration & $S_{\rm Cen-log}$ & $0.0206 \pm 0.0049$ & $0.0312 \pm 0.0084$ & $0.0299 \pm 0.0115$ \\
             & $S_{\rm Cen-log-simple}$ & $0.0213 \pm 0.0049$ & $0.0343 \pm 0.0102$ & $0.0288 \pm 0.0127$ \\
  \end{tabular}
\end{table}

%%%%%%%%%%%%%%%%%%%%%%%%%%%%%%%%%%%%%%%%%%%%%%%%%%%%%%%%%%%%%%%%%%%%%%%%%%%%%%%
%%%%%%%%%%%%%%%%%%%%%%%%%%%%%%%%%%%%%%%%%%%%%%%%%%%%%%%%%%%%%%%%%%%%%%%%%%%%%%%

\end{document}